\documentclass[onefignum,onetabnum]{siamart220329}

\usepackage[nobiblatex]{xurl}

\usepackage{multirow}
\DeclareFontFamily{U}{matha}{\hyphenchar\font45}
\DeclareFontShape{U}{matha}{m}{n}{
      <5> <6> <7> <8> <9> <10> gen * matha
      <10.95> matha10 <12> <14.4> <17.28> <20.74> <24.88> matha12
      }{}
\DeclareSymbolFont{matha}{U}{matha}{m}{n}
\DeclareMathSymbol{\odiv}         {2}{matha}{"63}
\usepackage{lipsum}
\usepackage{amsfonts}
\usepackage{graphicx}
\usepackage{epstopdf}
\usepackage{algorithmic}
\ifpdf
  \DeclareGraphicsExtensions{.eps,.pdf,.png,.jpg}
\else
  \DeclareGraphicsExtensions{.eps}
\fi

\def\mb{\mathbf}
\makeatletter

\newcommand{\Rmnum}[1]{\expandafter\@slowromancap\romannumeral #1@}
\makeatother

\mathchardef\mhyphen="2D

\newsiamremark{remark}{Remark}
\newsiamremark{hypothesis}{Hypothesis}
\crefname{hypothesis}{Hypothesis}{Hypotheses}
\newsiamthm{claim}{Claim}

\title{Consensus-based adaptive sampling and approximation for high-dimensional energy landscapes \thanks{Submitted to the editors DATE.
\funding{The work is supported in part by the
National Science Foundation under Grant DMS-2110981 and DMS-2143739, and the
ACCESS program through allocation MTH210005.}}}

\author{Liyao Lyu\thanks{Department of Computational Mathematics, Science \& Engineering}, {Michigan State University}, {{428 S Shaw Ln}, {East Lansing}, {48824}, {MI}, {USA}
  (\email{lyuliyao@msu.edu}).}
\and Huan Lei\thanks{Department of Computational Mathematics, Science \& Engineering}, {Michigan State University}, {{428 S Shaw Ln}, 
Department of Statistics \& Probability}, Michigan State University, Wells Hall, 619 Red Cedar Road, {East Lansing}, {48824}, {MI}, {USA}.
  (\email{leihuan@msu.edu}).}

\usepackage{amsopn}

\begin{document}

\maketitle

% REQUIRED
\begin{abstract}
We present a consensus-based framework that unifies phase space exploration with posterior-residual-based adaptive sampling for surrogate construction in high-dimensional energy landscapes.
Unlike standard approximation tasks where sampling points can be freely queried, systems with complex energy landscapes such as molecular dynamics (MD) do not have direct access to arbitrary sampling regions due to the physical constraints and energy barriers; the surrogate construction further relies on the dynamical exploration of phase space, posing a significant numerical challenge. We formulate the problem as a minimax optimization that jointly adapts both the surrogate approximation and residual-enhanced sampling. The construction of free energy surfaces (FESs) for high-dimensional collective variables (CVs) of MD systems is used as a motivating example to illustrate the essential idea. Specifically, the maximization step establishes a stochastic interacting particle system to impose adaptive sampling through both exploitation of a Laplace approximation of the max-residual region and exploration of uncharted phase space via temperature control. The minimization step updates the FES surrogate with the new sample set. Numerical results demonstrate the effectiveness of the present approach for biomolecular systems with up to 30 CVs. While we focus on the FES construction, the developed framework is general for efficient surrogate construction for complex systems with high-dimensional energy landscapes.

% One essential problem in quantifying the collective behaviors of molecular systems lies in the accurate construction of free energy surfaces (FESs). The main challenges arise from the prevalence of energy barriers and the high dimensionality. Existing approaches are often based on sophisticated enhanced sampling methods to establish efficient exploration of the full phase space. On the other hand, the collection of optimal sample points for the numerical approximation of FESs remains largely under-explored, where the discretization error could become dominant for 
% systems with a large number of collective variables (CVs).
% We propose a consensus sampling based approach by reformulating the construction as a minimax problem which simultaneously optimizes the function representation and the training set. In particular, the maximization step establishes a stochastic interacting particle system to achieve the adaptive sampling of the max-residue regime by modulating the exploitation of the Laplace approximation of the current loss function and the exploration of the uncharted phase space; the minimization step updates the FES approximation with the new training set. By iteratively solving the minimax problem, the present method essentially achieves an adversarial learning of the FESs with unified tasks for both phase space exploration and posterior error-enhanced sampling. We demonstrate the method by constructing the FESs of molecular systems with a number of CVs up to $30$. 
\end{abstract}

% REQUIRED

% REQUIRED

\section{Introduction}

Many problems in computational science rely on the efficient approximation of physical quantities in systems governed by high-dimensional energy landscapes, where direct sampling and evaluation are constrained by complex system dynamics. Canonical examples include the computation of transition paths in gradient systems \cite{E_Ren_PRB_2002, e2010transition}, the construction of coarse-grained molecular dynamics (MD) models \cite{noid2013perspective}, and uncertainty quantification under energy-induced probability measures \cite{lelievre2016partial}. Unlike the standard surrogate construction tasks where the sampling points can be freely queried, these problems further rely on efficient exploration and sampling over the phase space, where the thermodynamically accessible regions are often unknown \emph{a priori}. There are two essential challenges: (\Rmnum{1}) the prevalence of energy barriers, which makes direct sampling inefficient and prone to getting trapped in local minima; various enhanced sampling strategies are often required; (\Rmnum{2}) the high dimensionality of the surrogate model, which generally requires a large number of samples and motivates adaptive sampling strategies based on the approximation error of the target quantity. In practice, the efficient surrogate construction should account for both enhanced sampling with the complex energy landscapes and the residual-based adaptivity in constrained phase space --- yet simultaneously addressing both remains a nontrivial and open computational challenge.

In this study, we aim to develop a unified approach that enables both efficient phase space exploration and residual-based adaptive sampling for constrained phase space learning problems. As a motivating application, we consider the construction of free energy surfaces (FESs) with respect to a set of collective variables (CVs) for MD systems. While this is a long-standing problem in computational science, the accurate construction of high-dimensional FESs remains a difficult task due to the aforementioned two challenges. Most existing methods primarily target the first challenge based on various importance sampling strategies to overcome energy barriers, such as umbrella sampling~\cite{torrie1977nonphysical}, histogram reweighting~\cite{Kumar_Kollman_JCC_1992}, metadynamics~\cite{laio2002escaping,grafke2024metadynamics}, variationally enhanced sampling~\cite{valsson2014variational,shaffer2016enhanced,bonati2019neural}, and adaptive biasing force ~\cite{Darve_Pohorille_JCP_2001,darve2008adaptive,lelievre2008long,chipot2011enhanced}. Alternatively, temperature-accelerated  \cite{maragliano2006temperature}
and adiabatic dynamics \cite{rosso2002use,abrams2008efficient} introduce an extended dynamics of the CVs with an artificially high temperature to facilitate the phase space exploration. Despite their broad applications, the computational efficiency of these methods generally degrades as the number of CVs increases. More importantly, these methods do not explicitly address the second challenge. The residual error is not incorporated into the adaptive sampling and FES construction, which limits their effectiveness in high-dimensional problems. 

From a different perspective, several approaches related to the second challenge have been developed based on adaptive sampling \cite{tang2022adaptive, yu2022gradient, tang2023adversarial, Gao_Yan_SIAM_2023, Gao_Wang_JCP_2023, Jiao_Li_arxiv_2023, Han_Zhou_Stringnet_arxiv_2024} and adversarial learning \cite{zang2020weak, bao2020numerical, ZengCPINN_arxiv_2022} for solving high-dimensional partial differential equations (PDEs). The essential idea is to introduce certain residual-based distributions or weak formulations, where new collocation points or test basis functions can be adaptively updated during the training process. While they have shown promising results for high-dimensional PDEs, these methods rely on the free query of new sample points and their residual error within the domain. As such, they cannot be directly applied to the present problem, where the global residual error is unknown \emph{a priori}. In particular, the new sample points cannot be freely placed within the phase space but need to be navigated through dynamical exploration of the thermodynamically accessible regions. Alternatively, the reinforced dynamics (RiD) \cite{zhang2018reinforced} (see also Ref.~\cite{van2023hyperactive}) proposes using the uncertainty indicator as a proxy for the residual error to bias MD simulations, which, however, relies on calculating the standard deviation of the predictions from an ensemble of neural network (NN) surrogates trained on the same dataset. Moreover, the phase space exploration is constrained by the underlying energy landscape, which typically requires small time steps due to the stiffness and roughness of the MD potential function.

To address the above two challenges, we present a consensus-based adaptive sampling (CAS) method to efficiently construct surrogates within high-dimensional energy landscapes with applications to FES construction in complex MD models. A unique feature is that the present method enables gradient-free residual-based adaptive sampling such that the FES approximation and the phase-space exploration can be simultaneously optimized.  
The method is formulated as a minimax optimization problem. The max-problem seeks a residual-based distribution to establish adaptive sampling in the vicinity of the explored phase space regime, while the min-problem optimizes the FES surrogate based on the new samples.
For the maximization step, we emphasize that the establishment of the residual-based distribution is only formal; the analytical form of the distribution is unknown and the value at an arbitrary point can not be directly obtained. As a result, common sampling approaches based on Markov Chain Monte Carlo (MCMC) \cite{robert2004monte} and Langevin dynamics \cite{Gareth_Richard_Bernoulli_1996} are not applicable, as explained in more detail in Sec. \ref{sec:min_max}. Instead, we establish a consensus-based sampling \cite{carrillo2022consensus} (see also Refs. 
\cite{carrillo2018analytical, jingrun2022Consensus}) in the form of a stochastic interacting particle system governed by a McKean stochastic differential equation (SDE).
The gradient-free nature enables us to collect new samples adaptive to the local residual without the analytical form of the target distribution. 
Specifically, a quadratic potential is adaptively constructed to probe the local max-residual regime by exploiting the Laplace approximation under a low-temperature limit. Meanwhile, a coherent noise term is introduced to efficiently explore the full CV space under a high-temperature limit and yield the updated sampling points used for the subsequent minimization step.  

The present iterative procedure achieves adversarial learning of the FES pertaining to the thermodynamically important regions. In contrast to existing adaptive sampling methods for PDE solvers \cite{tang2022adaptive, tang2023adversarial, Gao_Yan_SIAM_2023}, the present method does not rely on the free query of arbitrary sample points. Instead, it enables us to establish a dynamical exploration of the phase space along with the surrogate construction. Moreover, unlike the reinforced dynamics~\cite{zhang2018reinforced} that navigates the sampling points through the biased MD simulations, the present sampling dynamics is governed by a smooth quadratic potential irrespective of the roughness of the underlying energy landscape, enabling much larger time steps and more efficient phase-space exploration. We demonstrate the effectiveness of the proposed method by constructing the FES of biomolecule systems involving up to $30$ CVs. Fig. \ref{fig:flow_chart} sketches the workflow of the proposed method.

\begin{figure}
\centering
\includegraphics[width=0.8\textwidth]{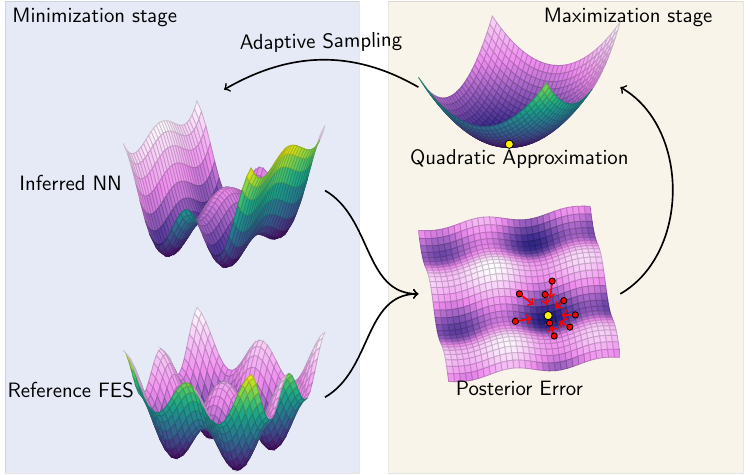}
\caption{The workflow of the present CAS-based method. In the minimization step, given a collection of sampling points, the reference force $F(\mb z)$ (i.e., the gradient of the underlying FES) can be calculated using the restrained molecular dynamics; a comparison with the force inferred from the NN surrogate $A_{\mathcal{N}}(\mb z)$ yields the loss function $L_{\mathcal{N}}(\mb z) = \vert \nabla_{\mb z} A_{\mathcal{N}}(\mb z) + F(\mb z)\vert^2 $. In the maximization step, the loss function determines a residual-based distribution with entropy regularization $q^{\ast}(\mb z) \propto \exp(\kappa_h\mathcal{L_N}(\mb z))$. An interacting particle system following the McKean SDE is used to achieve adaptive sampling of the max-residual region by modulating the exploitation of the Laplace approximation of $q^{\ast}(\mb z)$ and the exploration of the uncharted phase space. 
The FES can be accurately reconstructed after several iterations of the minimization and maximization step.}
\label{fig:flow_chart}
\end{figure}

% % The outline is not required, but we show an example here.
% The paper is organized as follows. Our main results are in
% \cref{sec:main}, our new algorithm is in \cref{sec:alg}, experimental
% results are in \cref{sec:experiments}, and the conclusions follow in
% \cref{sec:conclusions}.

\section{Methods}
\label{sec:model}
\subsection{Free energy and mean forces}
We consider a full model with micro-scale coordinates $\mathbf r \in \mathbb{R}^N$ whose dynamics is governed by potential $U(\mathbf{r}): \mathbb{R}^N \rightarrow \mathbb{R}$ under temperature $T$. Suppose we are interested in CVs $\mathbf{s}(\mathbf{r}): \mathbb{R}^N \to \Gamma$ with $\Gamma \subset  \mathbb{R}^M$, the FES $A(\mathbf{z})$ of the CVs is defined by 
\begin{equation}
    \begin{aligned}
        A(\mathbf{z}) = -\frac{1}{\beta}\ln \rho(\mathbf{z}),
    \end{aligned}
\end{equation}
where $\beta = 1/k_B T$ is the inverse of the thermal temperature, 
\begin{equation}
    \begin{aligned}
    \rho(\mathbf{z}) = \frac{1}{Z}\int \exp{(-\beta U(\mathbf{r}))}\delta(\mathbf{s}(\mathbf{r})-\mathbf{z}) \mathrm{d} \mathbf{r}
    \end{aligned}
\end{equation}
is the marginal probability density function (PDF) for $\mathbf{s}(\mathbf{r}) = \mathbf z$, $\delta (\cdot)$ represents the Dirac delta function and $Z = \int \exp{(-\beta U(\mathbf{r}))} \mathrm{d} \mathbf{r}$ is the partition function; we refer to Ref. \cite{stoltz2010free} for details. For high-dimensional CVs, direct estimation of $\rho(\mathbf z)$ often becomes numerically challenging. An alternative approach is to fit the mean force $\mathbf F(\mathbf z):= -\nabla A(\mathbf z)$ at various sample points, which can be estimated via the restrained MD \cite{allen2017computer} by introducing a harmonic term into the full potential, i.e.,
\begin{equation}
U_k(\mathbf{r},\mathbf{z}) = U(\mathbf{r}) +\frac{k}{2}(\mathbf{s}(\mathbf{r})-\mathbf{z})^\top(\mathbf{s}(\mathbf{r})-\mathbf{z}),
\label{eq:U_bias}
\end{equation}
where $k$ represents the magnitude of the restrained potential. 
As shown in Ref~\cite{maragliano2006temperature}, the mean force can be computed by $\mathrm{F}(\mathbf{z})=\lim_{k\to \infty}\mathrm{F}^k(\mathbf{z})$, where $\mathrm{F}^k(\mathbf{z})$ is defined by
\begin{equation}
\mathrm{F}^k(\mathbf{z})=\frac{1}{Z_k(\mathbf{z})}\int k(\mathbf{s}(\mathbf{r})-\mathbf{z})\exp{(-\beta U_k(\mathbf{r},\mathbf{z}))} \mathrm{d} \mathbf{r},
\end{equation}
%where $Z_k(\mathbf{z})=\int \exp{(-\beta U_k(\mathbf{r},\mathbf{z}))}\mathrm{d} \mathbf{r}$. 
and can be sampled as the first-moment estimation. 
%Compared with the direct estimation of the PDF $\rho(\mathbf z)$ from a histogram, matching the mean force $-\nabla A(\mathbf z)$ only requires the local first-moment estimation and therefore is used in this study. 

In principle, given a collection of sample points, $A(\mathbf z)$ can be re-constructed (up to a constant) by matching the mean force $-\nabla A(\mathbf z)$ at the individual points. However, for MD systems with complex energy landscape, the sampling over the phase space could be highly non-trivial due to the prevalence of energy local minima; the training set is often determined \emph{a prior} as pre-selected points or in a greedy manner. 
As the number of CVs increases, the empirical random samples often introduce pronounced discretization error. To efficiently construct $A(\mathbf z)$ in the thermodynamically accessible region, it is desirable to simultaneously optimize the training set and the FES surrogate approximation through certain adaptive sampling based on the \emph{posterior} residual error. This motivates the present method illustrated below.

%A DNN $\mathcal N(\mathbf{z})$ is employed as an ansatz for the FES. The parameters in DNN can be obtained by solving the equation
% \begin{equation}
% \label{equ:pde}
% \|\nabla_\mathbf{z} \mathcal N(\mathbf{z}) - \mathrm{F}(\mathbf{z}))\|=0.
% \end{equation}
% For simplicity, the DNN in this work are all fully-connected DNN model. 

\subsection{Min-Max formulation}\label{sec:min_max}
Let $A_{\mathcal N}(\mathbf{z})$ denote the NN surrogate of the FES $A(\mathbf z)$, which is parameterized by minimizing the loss function 
{\color{black}
\begin{equation}
\label{equ:pde}
\mathcal{L}_\mathcal{N}(\mathbf{z}) = \left \vert \nabla_\mathbf{z} A_{\mathcal N}(\mathbf{z}) + \mathrm{F}(\mathbf{z})\right \vert ^2
\end{equation}
}
for $\mathbf z \in \Gamma$. 
% A straightforward approach would involve simulating the evolution of $\mathbf{r}(t)$ and gathering information at $\mathbf{z}=\mathbf{s}(\mathbf{r}(t))$. However, this strategy loses effectiveness when the system displays metastability, characterized by local minima that trap the system.
To solve the minimization problem, we introduce a sampling distribution $q(\mathbf z)$
% , where $W$ is a function space that defines a proper constraint of the PDF $q(\mathbf z)$ on $\Gamma$
and define the weighted loss 
\begin{equation}
\label{equ:weak}
(\mathcal{L}_\mathcal{N},q) = \int_\Gamma \mathcal{L}_\mathcal{N}(\mathbf{z}) q(\mathbf{z}) \mathrm{d}\mathbf{z}.
\end{equation}
A desired distribution intends to maximize the discrepancy in the dataset for the current NN surrogate $A_{\mathcal N}(\mathbf{z})$. Accordingly, we define the maximization problem as
\begin{equation}
J[\mathcal{L}_\mathcal{N}]=\max_q (\mathcal{L}_\mathcal{N},q).
\end{equation}
%If $\mathcal{N}$ is a solution of \eqref{equ:pde}, then $G[\mathcal{L}_\mathcal{N}]=0$.
Since $(\mathcal{L}_\mathcal{N},q)$ is always non-negative, a good approximation of the original free energy surface $A_{\mathcal N}(\mathbf{z})$ (up to a constant) can be obtained by solving the following problems
\begin{equation}
\label{equ:min-max}
\min_{A_\mathcal{N}}\max_q (\mathcal{L}_\mathcal{N},q).
\end{equation}
% \begin{theorem}
% \label{thm:eq1}

\begin{proposition}
\label{thm:eq1}
Assuming that there exists a solution of $\mathcal{L}_\mathcal{N}(\mathbf{z}) = 0$ for $\mathbf z \in \Gamma$, then $A_{\mathcal{N}}^{\ast}$ is a solution if and only if it solves \eqref{equ:min-max}.
\end{proposition}
\begin{proof}
Suppose $A_{\mathcal{N}}^{\ast}$ is the solution for $\mathcal{L}_\mathcal{N}(\mathbf{z}) = 0$, it satisfies $(\mathcal{L}_{\mathcal{N}^{\ast}},q)=0$ for any $q$, i.e., $J[\mathcal{L}_{\mathcal{N}^{\ast}}] =0$. Therefore, $A_{\mathcal{N}}^{\ast}$ is a solution for the minimax problem \eqref{equ:min-max}. 
On the other hand, if $\hat{A}_{\mathcal{N}}$ is the minimizer for problem \eqref{equ:min-max} but not the solution for $\mathcal{L}_\mathcal{N}(\mathbf{z}) = 0$, then there exists $\hat{q} \in V$ such that $(\mathcal{L}_{\hat{\mathcal{N}}},\hat{q})>0$. However, $(\mathcal{L}_{\mathcal{N}^{\ast}},q)=0$ for all $q$, which contradicts the assumption that $\hat{A}_{\mathcal{N}}$ is the minimizer.
\end{proof}

Proposition \ref{thm:eq1} shows that direct construction of FES $A(\mathbf z)$ can be reformulated as an adversarial learning of an optimal solution $A_\mathcal{N}(\mathbf z)$ for the min-max problem \eqref{equ:min-max}. Accordingly, the training consists of two components: the minimization step optimizes the NN surrogate with the current training set; the maximization step explores the largest residual region for the current surrogate and essentially establishes an adaptive sampling of the training set based on the \emph{posterior} residual $\mathcal{L_N}(\mb z)$.  

To numerically solve the maximization problem, certain regularization needs to be introduced. Otherwise, the maximization problem will simply yield a delta measure, i.e., $\delta(\mathbf{z}-\mathbf{z}^*)$, where $\mathbf{z}^{\ast} = \arg\max \mathcal{L}_\mathcal{N}(\mathbf{z})$. Since the sampling needs to simultaneously identify the max-residual regime and explore the uncharted phase space, we introduce the entropy-based regularization \cite{wang2020reinforcement,gao2022state} (see also Refs. \cite{gulrajani2017improved,miyato2018spectral,tang2023adversarial} for gradient-based regularization), and the max-problem is reformulated by  
\begin{equation}
    \min_{q} \int (-\mathcal{L}_\mathcal{N}(\mathbf{z}) + \kappa_h^{-1} \ln q(\mathbf{z})) q(\mathbf{z}) \mathrm{d} \mathbf{z}.
\label{eq:entropy_regular}    
\end{equation}
This problem is convex with a unique global minimum at $q^*(\mathbf{z})=\exp(-\kappa_h \mathcal{L}_\mathcal{N}^-(\mathbf{z}))/Z^*$, where $Z^* = \int \exp(-\kappa_h \mathcal{L}_\mathcal{N}^-(\mathbf{z})) \mathrm{d} \mathbf{z}$  and $\mathcal{L}_\mathcal{N}^-(\mathbf{z}) = -\mathcal{L}_\mathcal{N}(\mathbf{z})$. The parameter $\kappa_h$ is a Lagrangian multiplier, balancing the weight between the max-residual concentration and the scope of exploration, and can be interpreted as the inverse of temperature in statistical physics. An elevated $\kappa_h^{-1}$ induces a distribution closer to a uniform distribution. Conversely, a diminished $\kappa_h^{-1}$ induces a distribution concentrated near the max-residual point.

Although $q^*(\mathbf{z})=\exp(-\kappa_h \mathcal{L}_\mathcal{N}^-(\mathbf{z}))/Z^*$ solves Eq. \eqref{eq:entropy_regular}, we emphasize that this solution is only formal. In particular, the analytical formulation is unknown since it further depends on the residual $\mathcal{L}_\mathcal{N}$ \eqref{equ:pde} and the exact FES $A(\mb z)$. Unfortunately, the numerical evaluation of $\nabla_{\mb z} A(\mb z^{\ast})$ at an arbitrary point relies on conducting restrained MD at $\mb z = \mb z^{\ast}$, which can be computationally expensive or even thermodynamically inaccessible. 
This constraint differs from conventional high-dimensional regression or PDE problems, where the residual can be freely queried within the domain.
Consequently, common sampling approaches based on the MCMC \cite{robert2004monte} and Langevin dynamics \cite{Gareth_Richard_Bernoulli_1996} can not be directly used. Specifically, the effectiveness of the MCMC method relies on the variance of the underlying distribution. However, most configurations in the MD phase space are thermodynamically inaccessible, resulting in large variance and extremely low sampling efficiency. On the other hand, the Langevin-based sampling method depends on $-\nabla_{\mb z} \ln q^*(\mathbf{z})$, which, in turn, relies on computing $\nabla_{\mb z} \mathcal{L_N^{-}}(\mb z)$ and essentially the gradient of the mean force, i.e., $\nabla_{\mb z} F(\mb z) = - \nabla_{\mb z} \otimes \nabla_{\mb z} A(\mb z)$ from restrained MD simulations. However, we can only sample the value of the mean force $F(\mb z)$ rather than its gradient.
%However the analytical formula of $q^{\ast}(\mathbf z)$ futher depends on the accurate FES, the direct sampling remains challenging for high-dimensional CVs. 

Inspired by the consensus-based sampling method \cite{carrillo2022consensus}, we propose a modified interacting particle system governed by the McKean SDE. The gradient-free nature enables us to efficiently construct the mean-field non-local conservative potential and the coherent noise without the evaluation of $\nabla_{\mb z} F(\mb z)$, and sample the particle distribution as an accurate approximation of $q^{\ast}(\mb z)$ illustrated below.

\subsection{Exploitation and exploration in the max-problem}
To approximate the target distribution $q^{\ast}(\mathbf z)$, particularly in the vicinity of the max-residual point $\mathbf{z}^*$, we exploit Laplace's principle in the large deviations theory, i.e., 
\begin{equation}
    \lim_{\kappa\to\infty}\left(-\frac{1}{\kappa}\log\left(\int \exp{(-\kappa f(\mathbf{z}))}\mathrm{d}\rho^*(\mathbf{z})\right)\right) = f(\mathbf{z}^*),
\end{equation}
which holds true for any compactly supported probability measure $\rho^{\ast}$, where $\mathbf{z}^{\ast} \in \text{supp}(\rho^*)$ uniquely minimizes the function $f$. 
This enables us to identify the max-residual point from a collection of samples $\left\{\mathbf z^{i}\right\}_{i=1}^{N_w}$ by the first-order moment  under the weighted density function $p(\mathbf{z})$, i.e.,  
 %$\mathbf{m} = \int \mathbf{z} p(\mathbf{z})\mathrm{d} \mathbf{z}, p(\mathbf{z}) = \frac{\exp{(-\kappa_l \mathcal{L}_\mathcal{N}^-(\mathbf{z}))}}{\int\exp{(-\kappa_l \mathcal{L}_\mathcal{N}^-(\mathbf{z}))} \mathrm{d} \mathbf{z}}$.
\begin{equation}
    \mathbf{m}_{\kappa_l} = \int \mathbf{z} \frac{p(\mathbf{z})}{\int p(\mathbf z') \rho(\mathrm d \mathbf z')}\rho(\mathrm{d}  \mathbf{z})  \approx 
    \sum_{i=1}^{N_w} \mathbf z^i \hat{p}(\mathbf z^i) 
    \quad \hat{p} (\mathbf z) =  \frac{\exp {(-\kappa_l \mathcal{L}_\mathcal{N}^-(\mathbf{z}))}}{\sum_{i=1}^{N_w}\exp {(-\kappa_l \mathcal{L}_\mathcal{N}^-(\mathbf{z}^i)})}
\label{eq:m_p}
\end{equation}
% \begin{equation}
%     \mathbf{m} = \int \mathbf{z} p(\mathbf{z})\mathrm{d} \mathbf{z} \quad p(\mathbf z) = 
%     \frac{1}{Z}  q(\mathbf z) \exp{(-\kappa_l \mathcal{L}_\mathcal{N}^-(\mathbf{z}))}
%     \approx \frac{\exp {(-\kappa_l \mathcal{L}_\mathcal{N}^-(\mathbf{z}^i))}}{\sum_{i=1}^{N_w}\exp {(-\kappa_l \mathcal{L}_\mathcal{N}^-(\mathbf{z}^i)})},
% \label{eq:m_p}
% \end{equation}
% \begin{equation}
%     \mathbf{m} = \int \mathbf{z} p(\mathbf{z})\mathrm{d} \mathbf{z} \quad p(\mathbf z) = 
%     \frac{q(\mathbf z)\exp{(-\kappa_l \mathcal{L}_\mathcal{N}^-(\mathbf{z}))}}{\int q(\mathbf z)\exp{(-\kappa_l \mathcal{L}_\mathcal{N}^-(\mathbf{z}))} \rm{d}\mathbf z}
%     \approx \frac{\exp {(-\kappa_l \mathcal{L}_\mathcal{N}^-(\mathbf{z}^i))}}{\sum_{i=1}^{N_w}\exp {(-\kappa_l \mathcal{L}_\mathcal{N}^-(\mathbf{z}^i)})},
% \label{eq:m_p}
% \end{equation}
% \begin{equation}
%     \mathbf{m} = \sum_i^{N_{w}} \mathbf{z}^i \hat{p}(\mathbf{z}^i) \quad \hat{p}(\mathbf{z}) =\frac{\exp {(-\kappa_l \mathcal{L}_\mathcal{N}^-(\mathbf{z}^i))}}{\sum_i^{N_w}\exp {(-\kappa_l \mathcal{L}_\mathcal{N}^-(\mathbf{z}^i)})},
% \label{eq:m_p}
% \end{equation}
where $\kappa_l^{-1}$ represents a low temperature limit, $p(\mb z) = \exp(-k_l \mathcal{L_N^-}(\mb z))$ and $ \rho  = \frac{1}{N_w}\sum_{i=1}^{N_w}\delta_{\mathbf{z}^i}$ is the empirical measure.  
% However, the integration is subject to the so-called curse of dimensionality as the number of CVs increases.
In order to sample the empirical measure towards $q^*$, we treat each sampler $\mathbf{z}^i$ as a random walker $\mathbf{z}^i_t$ governed by the following McKean SDE 
\begin{equation}
\label{equ:langevin}
%\dot{\mathbf{z}}^i_t = - \frac{1}{\gamma}\nabla_\mathbf{z} G(\mathbf{z}^i_t) + \sqrt{\frac{2}{\kappa_h\gamma}} \mathbf\xi_i(t), \quad i = 1, \cdots, N_w, 
\mathrm d \mathbf{z}^i_t = - \frac{1}{\gamma}\nabla_\mathbf{z} G(\mathbf{z}^i_t) \mathrm d t + \sqrt{\frac{2}{\kappa_h\gamma}} \mathrm d \mb W^{i}_t, \quad i = 1, \cdots, N_w, 
\end{equation}
where $G(\mathbf{z}_t)=\frac{1}{2} (\mathbf{z}_t-\mathbf{m}_{\kappa_l,t})^\top V_{\kappa_l,t}^{-1}(\mathbf{z}_t-\mathbf{m}_{\kappa_l,t})$ denotes an adaptively constructed mean-field conservative potential function.
The formulations of $\mathbf{m}_{\kappa_l,t}, V_{\kappa_l,t}$ are specified in \eqref{equ:mV} with the rationale discussed in the next section. Consequently, $G (\mathbf{z}_t)$ navigates the random walkers (i.e., individual particles) towards $\mathbf{m}_{\kappa_l,t}$, which represents the region of large residual error. 
%It is crucial to underscore that $\mathbf{m}_t$ is dictated by prior sampling. This characteristic renders $\mathcal{A}(\mathbf{z}_t)$ both adaptable and reactive to the existing state of the system.
%
The second term in Eq. \eqref{equ:langevin} represents a noise term where $\gamma$ represents the friction coefficient and 
$\mb W^{i}_t$ represents the standard $M\mhyphen$dimensional Brownian motion. 
%$\mathbf{\xi}(t)$ represents the standard Gaussian white noise characterized by zero mean and covariance $\mathbb{E}[\mathbf{\xi}_i(t)\mathbf{\xi}_j(t')] = \delta_{ij}\delta(t-t')$.  

%This distinctive attribute endows the dynamic system with the advanced capacity to probe into previously uncharted regions, thereby encouraging a more extensive exploration.

The coupling of the conservative and stochastic terms maintains a relatively high temperature $\kappa_h^{-1}$. Also, a large friction coefficient  $\gamma$ is applied such that the distribution of walkers $q_t(\mathbf{z})$ is almost always Gaussian during the evolution. As shown in Prop. \ref{prop:invariant_dist}, by properly choosing the form of $V_{\kappa_l,t}$,
the distribution $q_t(\mathbf{z})$ converges to  $\propto \exp(-\kappa_h G(\mathbf{z}))$ characterized by $\mathbf{m}_{\kappa_l, \infty}$ and $V_{\kappa_l, \infty}$.
Accordingly, the balance between exploitation and exploration is controlled using two temperatures $\kappa_l^{-1}$ and $\kappa_h^{-1}$. As $\kappa_l^{-1}$ decreases, the distribution concentrates near the max-residual points, reflecting the role of exploitation. Conversely, as $\kappa_h^{-1}$ increases, the distribution smoothens progressively, enhancing the exploration of the uncharted regions.

\begin{remark}
The present sampling dynamics \eqref{equ:langevin} takes a different form from the one proposed in Ref. \cite{carrillo2022consensus}. Specifically, the conservative potential $G(\mb z)$ is constructed by both the first and second moment which enables us to conveniently characterize $\mathcal{L_N}(\mb z)$ near the max-residual region with Laplace’s approximation, whereas it is only determined by the first moment in Ref. \cite{carrillo2022consensus}. Also, $\kappa_h$ is introduced as an independent parameter to modulate the exploration of the phase space. In principle, the consensus dynamics in Ref. \cite{carrillo2022consensus} could be also used for adaptive sampling with the proper choice of the multiplicative noise term. We proposed the modified form \eqref{equ:langevin} such that the sampling parameters can be chosen with a clear physical interpretation.
\end{remark}

\subsection{Convergence analysis}
In this subsection, we analyze the long-time behavior of the sampling dynamics \eqref{equ:langevin}. We show that the particle distribution converges exponentially fast to a steady state as an approximation of the target residual-based distribution $q^*(\mathbf{z})\propto \exp(-\kappa_h \mathcal{L}_\mathcal{N}^-(\mathbf{z}))$ under mild conditions. 
%we show that the sampling dynamics \eqref{equ:langevin}  converges to a steady state as an approximation of the target residual-based distribution $q^*(\mathbf{z})\propto \exp(-\kappa_h \mathcal{L}_\mathcal{N}^-(\mathbf{z}))$ with the proper choices of 
Specifically, we choose $\mathbf{m}_{\kappa_l,t}$ and $V_{\kappa_l,t}$ in the form of 
\begin{equation}
\begin{aligned}
\label{equ:mV}
\mathbf{m}_{\kappa_l,t} = \mathcal M _{\kappa_l}(\rho_t) \quad
V_{\kappa_l,t}=  \mathcal V _{\kappa_l}(\rho_t),
\end{aligned}
\end{equation}
where $ \rho_t  = \frac{1}{N_w}\sum_{i=1}^{N_w}\delta_{\mathbf{z}^i_t}$ is the empirical measure and 
\begin{equation}
    \begin{aligned}
    \mathcal M _{\kappa_l}(\rho) &=  \int \mathbf z \frac{ p(\mathbf z)}{\int p(\mathbf z') \rho( \mathrm d \mathbf z' )} \rho( \mathrm d \mathbf z ) ,\\
        \mathcal V _{\kappa_l}(\rho) &= \kappa_t \int \mathbf (\mb z-\mathcal M _{\kappa_l}(\rho))\otimes(\mb z-\mathcal M _{\kappa_l}(\rho))  \frac{p(\mathbf z)}{\int p(\mathbf z') \rho( \mathrm d \mathbf z' )}\rho( \mathrm d \mathbf z ).
    \end{aligned}
\end{equation}
% where . 
%
% \begin{equation}
% \begin{aligned}
% \label{equ:mV}
% \mathbf{m}_{\kappa_l,t} &= 
% V_{\kappa_l,t}&= \kappa_t\sum_{i=1}^{N_{w}} (\mathbf{z}^i_t-\mathbf{m}_t)(\mathbf{z}^i_t-\mathbf{m}_t)^T \hat{p}(\mathbf{z}^i_t),
% \end{aligned}
% \end{equation}
% where $\hat{p}(\mathbf z)$ is defined by Eq. \eqref{eq:m_p}. 
In particular, we show that by choosing $\kappa_t = \kappa_l + \kappa_h$,  Eq. \eqref{equ:langevin} converges to the target distribution $q^*(\mathbf{z})$. 
% While the motivation from the continuous limit has been substantiated in prior studies \cite{carrillo2018analytical,carrillo2022consensus} concerning sampling and optimization procedures, it was deemed necessary to incorporate it, albeit in a modified form, into our adaptive sampling task.
\begin{proposition}
\label{prop:invariant_dist}
Suppose $\mathcal{L}_\mathcal{N}^-(\mathbf{z})$ takes a local quadratic approximation in form of 
$\frac{1}{2} (\mathbf{z}-\mathbf{\mu})^\top\Sigma^{-1}(\mathbf{z}-\mathbf{\mu})$. 
If the dynamics converge to an invariant distribution, then the stationary density is given by 
\begin{equation}
    q_\infty = \frac{\exp{(-\kappa_h \mathcal{L}_\mathcal{N}^-(\mathbf{z}))}}{\int \exp {(-\kappa_h \mathcal{L}_\mathcal{N}^-(\mathbf{z}))} \mathrm{d} \mathbf{z}},
\end{equation} by choosing $\kappa_t = \kappa_l + \kappa_h$.
\end{proposition}
%We refer to SI for the proof. 
\begin{proof}
Let $q_{\infty}(\mathbf z)$ denote the invariant distribution of Eq. \eqref{equ:langevin}. Then $q_{\infty} (\mathbf z)$ must be the invariant distribution of the following SDE 
\begin{equation}
\mathrm d \mathbf z = - \frac{1}{\gamma} V^{-1}_{\kappa_l, \infty}(\mathbf z - \mathbf m_{\kappa_l, \infty})  \mathrm d t + \sqrt{\frac{2}{\gamma \kappa_h} } \mathrm d \mb W_t ,
\label{eq:asym_Langevin}
\end{equation}
where $\mathbf m_{\kappa_l, \infty}$ and $\kappa_t^{-1} V_{\kappa_l, \infty}$ are the mean and the covariance matrix of the re-weighted density $\propto ~q_{\infty}(\mathbf z) e^{-\kappa_l \mathcal{L}_{\mathcal{N}}^{-}(\mathbf z)}$. With the fluctuation-dissipation relation for Eq. \eqref{eq:asym_Langevin}, we can show $q_{\infty} (\mathbf z)$ follows the Gaussian distribution with mean $\mathbf m_{\kappa_l, \infty}$ and covariance matrix $\kappa_h^{-1} V_{\kappa_l, \infty}$. 

Since $\mathcal{L}_{\mathcal{N}}(\mathbf z) = \frac{1}{2} (\mathbf z-\mathbf\mu)^\top\Sigma^{-1}(\mathbf z-\mathbf\mu)$ is quadratic, the re-weighted density of a Gaussian distribution $q(\mathbf z) \sim \mathcal{N}(\mathbf m, V)$ remains Gaussian, i.e., $$
q(\mathbf z) e^{-\kappa_l \mathcal{L}_{\mathcal{N}}^{-}(\mathbf z)}  \propto ~  \mathcal{N} ({\mathbf m}_{\kappa_l}, {V}_{\kappa_l}),
$$ where ${\mathbf m}_{\kappa_l}$ and ${V}_{\kappa_l}$ are defined by
\begin{equation}
\begin{split}
\mathbf m_{\kappa_l}(\mathbf m , V) &= (V^{-1} + \kappa_l \Sigma ^{-1}) ^{-1} (\kappa_l\Sigma^{-1}\mathbf\mu+V^{-1}\mathbf m),\\
V_{\kappa_l}(\mathbf m , V) &= (V^{-1} + \kappa_l \Sigma ^{-1}) ^{-1}.
\end{split}
\label{eq:m_v_reweight_Gauss}
\end{equation}
% In particular, we choose $\mathbf m = \mathbf m_{\kappa_l, \infty}$ and $V = \kappa_h^{-1} V_{\kappa_l, \infty}$, then Eq. \eqref{eq:m_v_reweight_Gauss} yields
Hence, the mean and covariance of the steady-state Gaussian distribution satisfies
{\color{black}\begin{equation}
\begin{aligned}
\mathbf m_{\kappa_l, \infty} &= (\kappa_h V_{\kappa_l, \infty}^{-1} + \kappa_l \Sigma ^{-1})^{-1}  (\kappa_l\Sigma^{-1}\mathbf\mu+\kappa_h V_{\kappa_l, \infty}^{-1} \mathbf m_{\kappa_l, \infty}),\\ 
V_{\kappa_l, \infty}  & = \kappa_t (\kappa_h V_{\kappa_l, \infty}^{-1} + \kappa_l \Sigma ^{-1}) ^{-1}.
\end{aligned}
\nonumber
\end{equation}}
It is easy to show that by choosing $\kappa_t = \kappa_l + \kappa_h$, $\mathbf m_{\kappa_l, \infty}$ and $V_{\kappa_l, \infty}$ recovers $\mathbf\mu$ and $\Sigma$, respectively, and the invariant distribution takes the form
\begin{equation}
q_{\infty}(\mathbf z) \sim \mathcal{N}\left(\mathbf\mu, \kappa_h^{-1} \Sigma\right).
\nonumber
\end{equation}

\end{proof}

We emphasize that Eqs. \eqref{equ:langevin} and \eqref{equ:mV} differ from the standard Langevin dynamics; the sampling relies on neither the explicit knowledge of the target distribution $q^*(\mathbf{z})\propto \exp(-\kappa_h \mathcal{L}_\mathcal{N}^-(\mathbf{z}))$ nor the numerical evaluation of $\nabla_{\mb z} q^*(\mathbf{z})$ for individual sampling points. 
The quadratic assumption of the loss function $\mathcal{L}_\mathcal{N}^-(\mathbf{z})$ is due to the fact that we are mainly interested in the thermodynamically accessible region near the max-residual point. Under the low-temperature limit, the local regime can be well characterized by the first and second moment following Laplace's principle.  %While $\mathbf m_{\kappa_l, t}$ identifies the extremal point, $V_{\kappa_l, t}$ recognizes the anisotropic nature among the different CVs.

Next, we show that, under appropriate conditions, the sampling dynamics \eqref{equ:langevin} converges to the target distribution exponentially fast. In particular, we 
consider the time discretization given by 
{\color{black}
\begin{equation}\label{equ:discretize}
    \mathbf{z}^i _{t+1} =  \mathbf{z}^i_{t}  - \frac{ \mathcal V_{\kappa_l}^{-1} (\rho_t)}{\gamma} \left(\mathbf{z}^i_t - \mathcal M _{\kappa_l} (\rho_t) \right) \delta t  + \sqrt{\frac{2\delta t}{\gamma \kappa_h }} \eta^i_t,
\end{equation}
}
where $\eta^i _t$ are independent $\mathcal{N}(0,\mathbb{I}_M)$ random variables and $\delta t$ is the time step. 
% $\rho_t = \sum_{i=1}^{N_w}\delta_{\mathbf{z}^i_t}$ is the empirical measure  and $\mathcal M_{\kappa_l}, \mathcal V_{\kappa_l}$ denote respectively the mean and variance for the reweighting of measure
\begin{lemma}
    If the initial law $\rho_0$ for \eqref{equ:discretize} is Gaussian, so is the law for any $t \in \mathbb N_{>0}$. Also, its evolution is characterized by the first and second moments $(\mathbf m _t , V_t )$ of $\rho_t$, which are governed by the recurrence relation
 {\color{black}   
    \begin{equation}
        \begin{aligned}
            \mathbf m _{t+1}  & = \left(\mathbb{I}_M-\frac{\delta t  V_{k_l}^{-1}(\mathbf m_t, V_t) }{\gamma}\right)\mathbf m_t  + \frac{\delta t  V^{-1}_{k_l}(\mathbf m_t, V_t) }{\gamma} \mathbf m_{k_l} (\mathbf m_t, V_t) ,\\
            V _{t+1}  & = \left(\mathbb{I}_M-\frac{\delta t  V_{k_l}^{-1}(\mathbf m_t, V_t) }{\gamma}\right)^2 V_t + \frac{2\delta t}{\gamma\kappa_h} ,
        \end{aligned}
    \end{equation}
}
    where $\mathbf m_{k_l} (\mathbf m, V) = \mathcal M _{\kappa_l} (g(\cdot;\mathbf m , V))$ and $V_{k_l} (\mathbf m, V) = \mathcal V _{\kappa_l} (g(\cdot;\mathbf m , V))$.
\end{lemma}
\begin{proof}
The time evolution of the law of the solution \eqref{equ:discretize} is governed by the following discrete-time dynamics on probability density:
{\color{black}
\begin{equation}\label{equ:discrete_time_distribution}
    \rho_{t+1}(\mathbf{z}) =  \int_{\mathbb R^M} g\left(\mathbf z; \mathbf z' -\frac{\delta t  V_{\kappa_l}^{-1} (\rho_t)}{\gamma} \left(\mathbf{z}' - \mathcal M _{\kappa_l} (\rho_t) \right) , \frac{2\delta t}{\gamma\kappa_h}  \right)  \rho_t (\mathbf z') \mathrm{d} \mathbf z',
\end{equation}
}
where $g(\mathbf z; \mathbf m,V) =  \frac{1}{\sqrt{(2\pi)^M \det (V)}} \exp(-\frac{1}{2} 
\left(\mathbf{z} - \mathbf m )^\top V^{-1} (\mathbf{z} - \mathbf m )  \right)$.
Since both \(\rho_t\) and \(g\) are Gaussian, the convolution in \eqref{equ:discrete_time_distribution} results in another Gaussian distribution. Therefore, \(\rho_t\) remains Gaussian for all \(t > 0\). 
The update rules for the mean and covariance can be derived by explicitly computing the first and second moments of \(\rho_{t+1}\) from the integral form \eqref{equ:discrete_time_distribution}.
\end{proof}

\begin{proposition}
    Assume that $\mathcal{L}_\mathcal{N}^-(\mathbf{z})$ takes a local quadratic approximation in form of $\frac{1}{2} (\mathbf{z}-\mathbf{\mu})^\top \Sigma^{-1}(\mathbf{z}-\mathbf{\mu})$ and initial condition $(\mathbf m_0,V_0) \in \mathbb R^M \times S^M_{++}$, where $S^M_{++}$ denotes the set of symmetric strictly positive definite matrices in $\mathbb R^{M\times M}$. Then the first and second moment of the empirical distribution $\mathbf m_t $ and $V_t$ converge to $\mu$ and $\kappa_h^{-1} \Sigma$ exponentially fast if $\delta t$ is sufficiently small.  To be more specific, we have 
    \begin{equation}
    \begin{aligned}
        \mathbf m_{t}  &= \mu +  \left(\mathbb{I}_M -  \frac{\delta t \kappa_t^{-1} \kappa_l \Sigma^{-1}}{\gamma}\right)^{t}\mathbf (\mathbf m_0-\mu)\\
        V_{t}  & = \kappa_h^{-1} \Sigma +  \Sigma\left(\mathbb I_M -\frac{2\delta t \kappa_l\Sigma^{-1}}{\gamma(\kappa_l + \kappa_h)} \right)^{t}\left(\Sigma^{-1}V_0 -  \kappa_h^{-1} \mathbb I_M \right).
    \end{aligned}
    \end{equation}
\end{proposition}
\begin{proof}
    For the first moment, we have the following recurrence relation
\begin{equation}
\begin{aligned}
    \mathbf m _{t+1}   &= \mathbf m_t  - \frac{\delta t \kappa_t^{-1} (V^{-1}_t +\kappa_l \Sigma^{-1}) }{\gamma}\left(\mathbf m_t - (V_t^{-1} + \kappa_l \Sigma ^{-1}) ^{-1} (\kappa_l\Sigma^{-1}\mathbf\mu +V^{-1}_t\mathbf m_t) \right)\\
    & = \left(\mathbb{I}_M -  \frac{\delta t \kappa_t^{-1} \kappa_l \Sigma^{-1}}{\gamma}\right)\mathbf m_t +\frac{\delta t\kappa_t^{-1}\kappa_l\Sigma^{-1}}{\gamma}    \mu.
\end{aligned}
\end{equation}
Therefore, we have that 
\begin{equation}
\begin{aligned}
    \mathbf m _{t+1} -   \mu &=   \left(\mathbb{I}_M -  \frac{\delta t \kappa_t^{-1} \kappa_l \Sigma^{-1}}{\gamma}\right)\mathbf (\mathbf m_t-\mu) \\
   & =   \left(\mathbb{I}_M -  \frac{\delta t \kappa_t^{-1} \kappa_l \Sigma^{-1}}{\gamma}\right)^{t+1}\mathbf (\mathbf m_0-\mu). 
\end{aligned}
\end{equation}
By \textcolor{black}{choosing} the timestep $\delta t\in\left(0,\frac{\gamma}{\kappa_t^{-1}\kappa_l e(\Sigma)^{-1}}\right)$, where $e(\Sigma)$ denote the smallest eigenvalue of $\Sigma$, the first moment will converge to $\mu$ exponentially fast. For the second moment, we have that 
\begin{equation}
        V_{t+1} =  \left(\mathbb{I}_M-\frac{\delta t \kappa_t^{-1} (V^{-1}_t +\kappa_l \Sigma^{-1})  }{\gamma}\right)^2 V_t + \frac{2\delta t}{\gamma\kappa_h}
\end{equation}
By assuming that $\delta t$ is small, we have the following approximation
\begin{equation}\begin{aligned}
   V_{t+1} &=  \left(\mathbb{I}_M-\frac{2\delta t \kappa_t^{-1} (V^{-1}_t +\kappa_l \Sigma^{-1})  }{\gamma}\right) V_t + \frac{2\delta t}{\gamma\kappa_h}\\
 & =   V_t  - \frac{2\delta t \kappa_l \Sigma^{-1} V_t}{\gamma(\kappa_l + \kappa_h)} + \frac{2\delta t\kappa_l}{\gamma\kappa_h(\kappa_l + \kappa_h)} \\
 & = V_t  - \frac{2\delta t \kappa_l}{\gamma(\kappa_l + \kappa_h)}\left(\Sigma^{-1}V_t -  \kappa_h^{-1} \mathbb I_M \right).
\end{aligned}
\end{equation}
Therefore, we have that 
\begin{equation}
    \begin{aligned}
        \Sigma^{-1}V_{t+1} -  \kappa_h^{-1} \mathbb I_M  &=  \left(\mathbb I_M -\frac{2\delta t \kappa_l\Sigma^{-1}}{\gamma(\kappa_l + \kappa_h)} \right)\left(\Sigma^{-1}V_t -  \kappa_h^{-1} \mathbb I_M \right)\\
        & = \left(\mathbb I_M -\frac{2\delta t \kappa_l\Sigma^{-1}}{\gamma(\kappa_l + \kappa_h)} \right)^{t+1}\left(\Sigma^{-1}V_0 -  \kappa_h^{-1} \mathbb I_M \right).
    \end{aligned}
\end{equation}
By choosing the timestep $\delta t\in\left(0,\frac{\gamma}{2 \kappa_t^{-1}\kappa_l e(\Sigma)^{-1}}\right)$, the second moment will converge to $\kappa_h^{-1} \Sigma$ exponentially fast.
\end{proof}
{\color{black}{
    
The convergence analysis in Proposition 2.3 and 2.5 assumes a local quadratic model 
$\mathcal L^-_{\mathcal N}(z)\approx \tfrac12(z-\mu)^\top\Sigma^{-1}(z-\mu)$, 
under which the reweighted law of a Gaussian remains Gaussian and yields closed-form updates of $(m_t,V_t)$.
In a general non-quadratic regime, the McKean dynamics~\eqref{equ:langevin} is driven only by the first two moments; hence it can be interpreted as a \emph{moment-closure} that evolves a Gaussian surrogate $q_t\approx\mathcal N(m_t,V_t)$.

If $\mathcal L^-_{\mathcal N}$ is $\ell$-strongly convex with bounded curvature on the region of interest,
i.e., $\ell I \preceq D^2\mathcal L^-_{\mathcal N}(z)\preceq u I$,the induced reweighting map admits a uniform covariance bounds of the form
$(V^{-1}+\kappa_l u I)^{-1}\preceq V_{\kappa_l}(m,V)\preceq (V^{-1}+\kappa_l \ell I)^{-1}$,
and the reweighted mean approaches the minimizer at the scale $O(\kappa_l^{-1/2})$; see, e.g.,  Lemmas~3.1--3.3 in~\cite{carrillo2022consensus}.
Consequently, in this convex regime, the method targets a unique Gaussian steady state
which can be viewed as a Laplace-type approximation of $q^*(z)$ near the max-residual \textcolor{black}{region}. For non-convex (e.g., multi-modal) residual landscapes, even though we cannot approximate $q^*(z)$ over the whole space at once, the sampling dynamics will explore at least one local minima in each iteration and solve the min-max problem iteratively to explore the full significant phase space.
}
}
\subsection{Practical stabilization modifications}
In this study, for the sake of computational efficiency, we further simplify $V$ by only considering the diagonal entries denoted as $\mathbf{v}$. Moreover, we utilize a moving average to ensure stable estimation
\begin{equation}
\begin{aligned}
\mathbf{m}_{t+1} &= \beta_1 \mathbf{m}_{t} + (1-\beta_1) \sum_{i=1}^{N_{w}}\mathbf{z}^i_t \hat{p}(\mathbf z^i_t),\\
\mathbf{v}_{t+1} &= \beta_2\mathbf{v}_{t} + (1-\beta_2)\kappa_t\sum_{i=1}^{N_{w}} (\mathbf{z}^i_t-\mathbf{m})\odot(\mathbf{z}^i_t-\mathbf{m}) \hat{p}(\mathbf z^i_t),
\end{aligned}
\end{equation}
where $\beta_1$ and $\beta_2$ are hyper-parameters. 
To maintain unbiased estimation, normalization is implemented as follows: $\mathbf{m} = \frac{\mathbf{m}_t}{1-\beta_1^t}$ and $\mathbf{v} = \frac{\mathbf{v}_t}{1-\beta_2^t}$.

Furthermore, we note that the kinetic processes of a molecular system are generally characterized by the local minima and saddle points of the FES. On the other hand, the regimes of high free energy are less relevant.  To accurately construct these thermodynamically accessible regimes, we modify the loss function as
\begin{equation}
\mathcal{L}_\mathcal{N}(\mathbf{z}) = \frac{ \left \vert\nabla_\mathbf{z} A_{\mathcal N}(\mathbf{z}) + \mathrm{F}(\mathbf{z})\right \vert^2}{\left \vert \mathrm{F}(\mathbf{z})\right \vert ^2+e},
\end{equation}
for all the biomolecule systems except the toy example of the 1D Rastrigin function. $e$ is a small value to regularize the denominator.

With the training samples obtained from the aforementioned maximization step, the NN surrogate is optimized using the Adam stochastic gradient descent method \cite{Kingma_Ba_Adam_2015} for the minimization step. The loss function of the updated $A_{\mathcal N}(\mathbf{z})$, in turn, navigates the consensus-based adaptive sampling for the updated maximization step. The min-max problem is solved iteratively to achieve comprehensive sampling of the full phase space.  A detailed algorithm is presented in Alg. \ref{alg:ces}.

\begin{algorithm}
\caption{Consensus-based adaptive sampling.}
%In this algorithm, we use coefficients $\beta_1$ and $\beta_2$ to approximate the moving averages $\mathbf{m}$ and $\mathbf{v}$. The time step $\alpha$, friction coefficients $\gamma$, and inverse temperature $\beta^*$ determine the dynamics of the sampling point on the posterior error space. Additionally, we set $\epsilon=0.1$ to increase the robustness of the algorithm in cases where $\mathbf{v}$ is too small. Note that we assume $V$ to be a diagonal matrix and perform element-wise computation of $(\mathbf{z}^i_{t}-\hat{\mathbf{m}})/\hat{\mathbf{v}}$ to improve computational efficiency.}
\begin{algorithmic}
\REQUIRE{Initial sampling point $\mathbf{z}^i_0$, for $i=1,\ldots,N_{w}$}
\REQUIRE{Initial NN parameter $\theta_0$}
\REQUIRE{The number of training iterations $N_{train}$}
\REQUIRE The number of data collected $N_{data}$ in each training iteration
\STATE{$j \gets 0, t \gets 0 $}
\STATE $T \gets \lceil \frac{N_{data}}{N_{w}} \rceil $
\WHILE{$j < N_{train}$}
\WHILE{$t \leq T$}
\STATE calculate the mean force $\mathrm{F}^i_t$ at $\mathbf{z}^i_t$
\STATE calculate the predicted force $\mathcal{F}_\theta(\mathbf{z}^i_t) = \nabla_\mathbf{z} A_{\mathcal{N}}(\mathbf{z}^i_t; \theta_j)$
\STATE $L^i \gets \mathcal{L}_\mathcal{N}(\mathbf{z}^i_t)$
\STATE $w^i \gets \frac{\exp{(\kappa_l L^i})}{\sum_i \exp{(\kappa_l L^i})}$
\STATE $\mathbf{m}_{t+1} \gets \beta_1 \mathbf{m}_{t} + (1-\beta_1) \sum_i\mathbf{z}^i_t w^i$ 
\STATE $\mathbf{v}_{t+1} \gets \beta_2 \mathbf{v}_{t}+ (\kappa_l+\kappa_h)(1-\beta_2)\sum_i(\mathbf{z}^i_t-\mathbf{m}_t)^2 w^i$
\STATE $\mathbf{m} \gets \frac{\mathbf{m}_{t+1}}{1-\beta_1^t}$
\STATE $\mathbf{v} \gets \frac{\mathbf{v}_{t+1}}{1-\beta_2^t}$
\STATE $\mathbf{z}^i_{t+1} \gets \mathbf{z}^i_t - \frac{\delta t}{\gamma} (\mathbf{z}_t^i-\mathbf{m})\odiv\mathbf{v}+ \sqrt{\frac{2\delta t}{\gamma\kappa_h}}\mathbf{\eta}^i_t, \eta^i_t \sim \mathcal{N}(0,1)$
\STATE $t \gets t+1 $
\ENDWHILE
\STATE Save the training dataset $\mathcal{D}_j=\{\mathbf{z}^i_t,\mathrm{F}^i_t\}_{t=0}^T$
\STATE Optimize $\theta_{j+1}$ using the generated training set $\mathcal{D}_l$ for $l=0,\ldots,j$.
\STATE $j \gets j+1 $
\ENDWHILE
\end{algorithmic}
\label{alg:ces}
\end{algorithm}

\subsection{Related methods}

Before we present the numerical results, we further discuss the differences between the present CAS method and the existing approaches. Since the standard sampling methods such as MCMC \cite{robert2004monte}, Langevin dynamics \cite{Gareth_Richard_Bernoulli_1996}, and generative model \cite{dinh2016density, tang2022adaptive} are not applicable, we focus on two major approaches based on the enhanced sampling such as the VES method \cite{bonati2019neural}  as well as the active learning such as the RiD method \cite{zhang2018reinforced}.

The VES method provides an efficient approach to establish the re-weighted sampling of the phase space. However, the surrogate construction of the FES $A(\mathbf z)$ further relies on the estimation of the PDF $\rho(\mathbf z)$ from the obtained reweighted samples, which could become computationally expensive for high-dimensional problems. In contrast, the present method circumvents the PDF estimation and directly constructs $A(\mathbf z)$ through the adaptive sampling of the mean force $F(\mb z) = -\nabla A(\mathbf z)$.

Furthermore, the present method differs from the RiD method in the following two aspects: (\Rmnum{1}) The present method imposes the sampling adaptivity by explicitly using the residual error $\mathcal{L_N}(\mathbf z)$ while the RiD method does not directly account for the approximation error. Instead, the RiD method trains a replica of NNs on the same sample set and uses the standard deviation of multiple NN predictions as an indirect measure of the $\mathcal{L_N}(\mathbf z)$. As a result, the sampling error of the mean force $F(\mathbf z)$ may lead to a consistent biased prediction among multiple NNs with a small standard deviation. (\Rmnum{2}) The RiD method drives the sampling dynamics with a biased MD potential $U(\mb r)$ similar to Eq. \eqref{eq:U_bias}. Accordingly, the time step needs to be sufficiently small due to the stiffness of $U(\mb r)$. In contrast, the dynamics of the random walkers for the present method is governed by a mean-field quadratic potential $G(\mb z)$ in Eq. \eqref{equ:langevin}. In particular, $G(\mb z)$ is decoupled from the MD potential $U(\mb r)$. This enables us to choose a much larger time step irrespective of the roughness of $U(\mb r)$ and achieve a more efficient exploration of the phase space.  
Below we compare the performance of the present method with these two methods in more detail. 

%
%The effect of the different ways of imposing sampling adaptivity between the RiD and present methods becomes even more pronounced in the following high-dimensional problems. 

\section{Numerical Results}
\label{sec:numerical example}
{\color{black}{
In this section, we illustrate the CAS method with several numerical examples. We refer to Appendix \ref{sec:parameter_sen}, \ref{sec:explore_CV} and \ref{sec:additional_results} for the parameter choices and additional numerical results. 
}}
% {\color{red}
% The number of walkers $N_w$ affects the coverage of the CV space and barrier crossing, and the stability of estimating the moment variables $(m_t,V_t)$ under reweighting.
% In practice, $N_w$ depends not only on the CV dimension $d$ but also on the complexity of the underlying free-energy landscape (e.g., multimodality/anisotropy) and computational budget.
% We therefore use a simple tiered strategy: $N_w=10$ for 2D problems, $N_w=20$ for 3D, and $N_w=64$ for higher-dimensional cases (9D and 30D), which we found to provide stable moment estimates and robust error decay in our tests. The low-temperature parameter $\kappa_l$ determines the sharpness of the reweighting toward low-residual regions (too small slows contraction; too large may reduce effective sample size), while the high-temperature parameter $\kappa_h$ controls barrier-crossing efficiency (too small may lead to trapping; too large may slow refinement). These two parameters are quiet rubust in a very large range. 
% A brief robustness study on Ala2, including variations of $N_w$, is reported in Appendix~\ref{sec:parameter_sen}.
% }

\subsection{One-dimensional Rastrigin function}
To illustrate the essential idea of the present method, we start with the one-dimensional Rastrigin function:

\begin{equation}
f(z) = z^2-\cos(2\pi z ), \quad z\in [-3,3].
\end{equation}

Instead of a neural network, we simply use a piecewise polynomial function $f_\theta (z)$ for this 1D problem. Accordingly, the residual is directly defined as $\vert f-f_{\theta_i}\vert$, where $\theta_i$ represents the fitting parameters for the $i\mhyphen$th iteration.
Initially, we set $f_{\theta_0}(z) \equiv 8$ with consistent boundary condition. We use the proposed sampling method with 10 walkers to estimate the first and second moments, yielding $m_1=1 \times 10 ^{-4}$ and $V_1=0.022$ and therefore the first and second derivative $f'(m_1) = 0$ and $f''(m_1) = V_1^{-1} = 40.97$. The obtained $m_1$ is very close to the true values of the max-residual point $z_1^{\ast} = 0$ and the second derivative $f''(z_1^{\ast})=41.47$. Accordingly, we add new data points near $x=1 \times 10 ^{-4}$ into the training set, which enables us to construct an updated approximation $f_{\theta_1}(z)$.  
%Figure \ref{fig:1D} shows the fitting curve and the residue error of $f_{\theta_1}(x)$ in orange lines. 
Similarly, with the approximation $f_{\theta_{i}}(z)$, we conduct the sampling process (i.e., the maximization step) and include the obtained samples near $m_{i+1}$ into the training set, yielding an updated approximation $f_{\theta_{i+1}} (z)$. As shown in Fig. \ref{fig:1D}, for each iteration, the sampling step can accurately pinpoint the max-residual region. 
Furthermore, as shown in Appendix \ref{app:1D}, the second moment $V_i$ yields a consistent estimation of the second derivative near the max-residual region. The underlying function $f(z)$ can be accurately constructed after $12$ iterations. 

\begin{figure}
\centering
\includegraphics[width=0.8\textwidth]{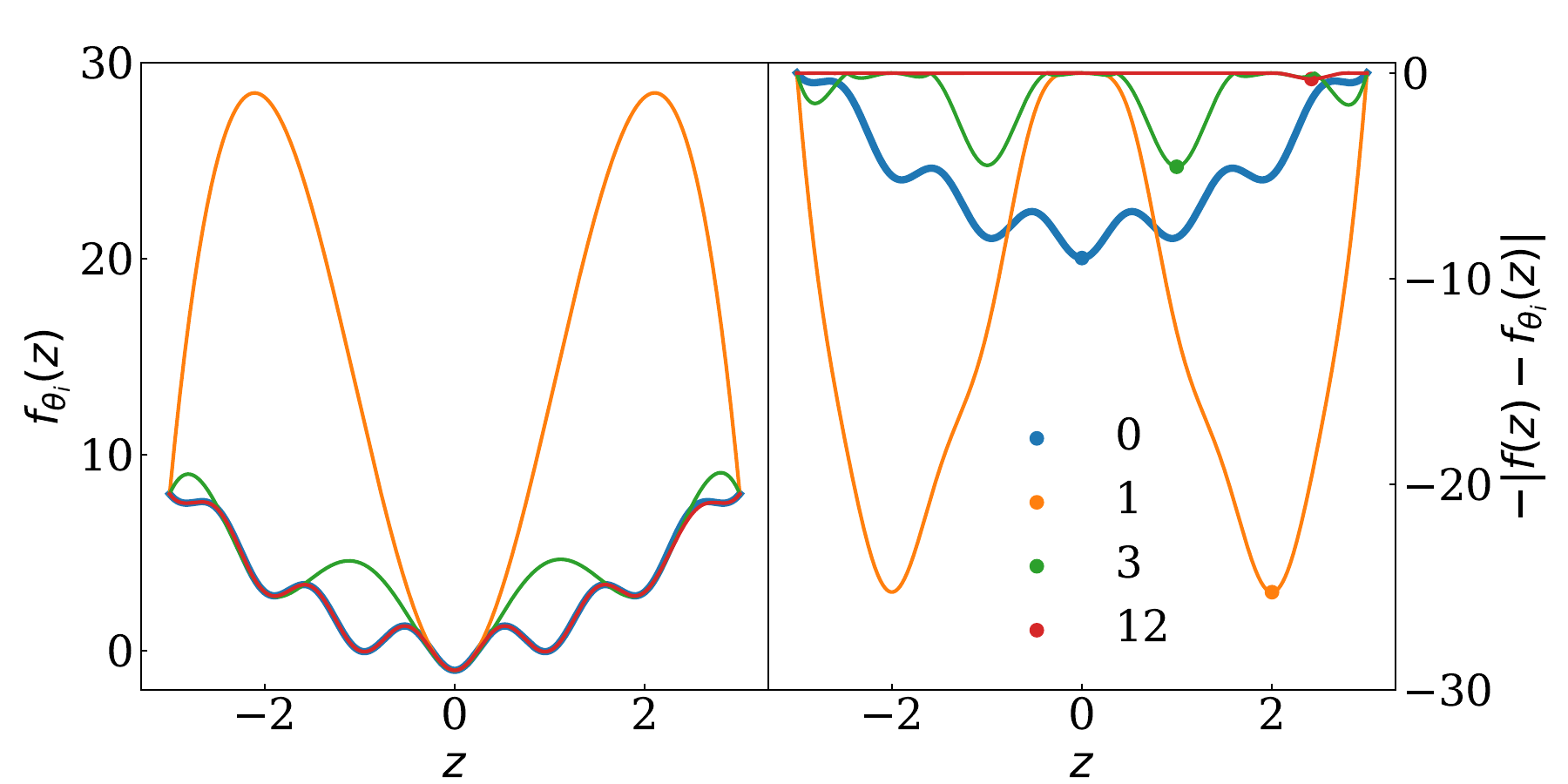}
\caption{Adaptive sampling and construction of the 1D Rastrigin function. Left: The reference function $f(z)$ and the constructed approximations $f_{\theta_i}(z)$ obtained at different iteration steps. The relative $l_2$ error is less than $6\times10^{-3}$ after $12$ iterations. Right: The residual function $|f(z)-f_{\theta_i}(z)|$. The symbols represent the locations identified by each adaptive sampling (i.e., the maximization) step where new points will be added for the next training (i.e., the minimization) step.}
\label{fig:1D}
\end{figure}

\subsection{Two-dimensional FES}

We use the alanine dipeptide (Ace-Ala-Nme), referred to as Ala2,  as a benchmark problem. The molecule is solvated in $383$ TIP3P water molecules similar to Ref. \cite{zhang2018reinforced}. The full MD system is simulated in a canonical ensemble under temperature $300$K using the Amber99-SB force field \cite{hornak2006comparison} with a time step $2$ fs. 
We refer to Appendix \ref{app:MD_setup} for the simulation details. %More details of simulation can be seen in \cite{zhang2018reinforced}

We choose the CVs as the two torsion angles $\phi$ (C, N, C$_\alpha$, C) and $\psi$ (N, C$_\alpha$, C, N). 
%Our primary objective is to benchmark the efficacy of our proposed method against established techniques, namely 
For comparison, we also construct the FESs using the  VES \cite{bonati2019neural} and RiD  \cite{zhang2018reinforced} method with the same setup parameters presented therein. 
For the present CAS method, we use $10$ walkers to explore the configuration space. The initial points of the walkers at the $k\mhyphen$th iteration are chosen to be the final states of the previous iteration. In the sampling stage, the inverse of the low and high temperatures are set to be $\kappa_l=10$ and $\kappa_h = 1$, respectively. For each sample point, restrained MD is conducted with $\kappa=500$  for $5000$ steps to compute the average force. The time step for the dynamics of the random walkers \eqref{equ:langevin} is set to be $0.1$.  
%Table~\ref{tab:2d_l2_error} shows the $l_2$ error of the FES constructed at each iteration. 
%
Fig. \ref{fig:ala2_2D} shows the FESs constructed by the three different methods and the reference obtained by the metadynamics \cite{laio2002escaping} using a long simulation time. As shown in Tab. \ref{tab:computational cost}, the present CAS method yields a smaller approximation error and meanwhile requires lower computational cost than the VES and the RiD method.

\begin{table}[h]
    \caption{The accuracy of the constructed 2D FES (the Ala2 molecule) and computational time (in hours, the same below) for the VES, RiD, and CAS methods. The \( l_2 \) and \( l_\infty \) error are computed up to \( 40 \) KJ/mol. The reference solution is constructed by the metadynamics. The simulation time of the CAS method is multiplied by $10$ since 10 random walkers are used.}
    \centering
    \begin{tabular}{|c|c|c|c|c|}
\hline
 \multirow{2}{*}{Method} &\multicolumn{2}{c|}{Accuracy} &   \multicolumn{2}{c|}{Time} \\
 \cline{2-5}
 ~ & $l_2$ error & $l_{\infty}$ error &Simulation & Train \\
 \hline
 VES  & $5.39$  & $21.03$ & \multicolumn{2}{c|}{$47.5$}   \\
\hline
 RiD & $3.15$ & $11.04$ & $17.98$ & $0.22$ (GPU)  \\
\hline
 \multirow{2}{*}{CAS}  & \multirow{2}{*}{$1.88$} & \multirow{2}{*}{$10.68$} & \multirow{2}{*}{$0.23 \times 10$} & $0.18$ (CPU) \\
 \cline{5-5}
  ~ & ~ & ~ &  ~   & $0.13$ (GPU)  \\
\hline
    \end{tabular}
    \label{tab:computational cost}
\end{table}

%\subsection{Two-dimensional FES}
\begin{figure}
    \centering
\includegraphics[width=8.7cm]{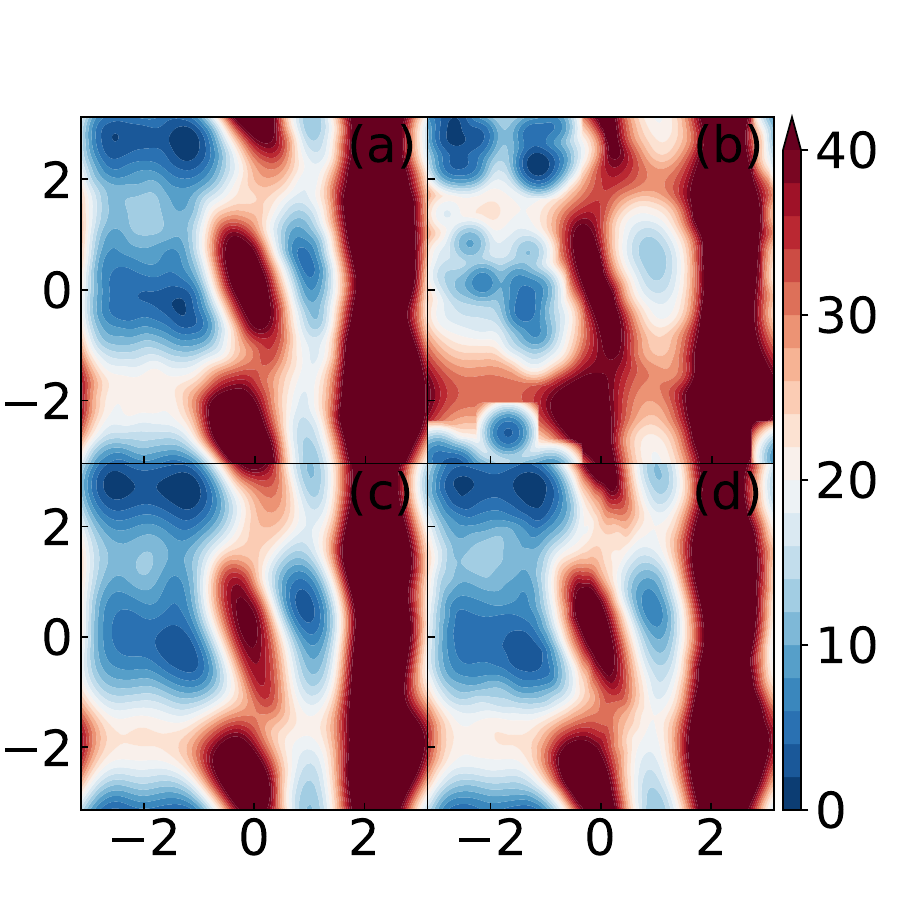}
    \caption{The 2D FES for the Ala2 molecule on the $\phi - \psi$ plane constructed by (a) Metadynamics (reference) (b) VES (c) RiD (d) CAS (the present method). The accuracy and the computational cost are shown in Table \ref{tab:computational cost}.}
    \label{fig:ala2_2D}
\end{figure}

%Table~\ref{tab:2d_l2_error} shows the $l_2$ error of the FES constructed at each iteration. 
%
% \begin{table}[]
%     \centering
%     \caption{The $l_2$ error and $l_\infty$ error of the present CAS method after different numbers of sampling iterations.}
%     \begin{tabular}{|c|ccccccc|}
%     \hline 
%     iteration & 1 & 2 & 3 & 4 & 5 & 6 & 7 \\
%     \hline 
%     $l_{\infty}$ error &  35.21 &   30.62& 22.38&  16.41&  14.83&  9.86&  10.68\\
%     \hline
%     $l_2$ error &  9.82 &  7.52& 6.84&  3.77 & 2.36 &  3.37 & 1.88   \\
%     \hline
%     \end{tabular}
%     \label{tab:2d_l2_error}
% \end{table}

\subsection{Three-dimensional FES}
Next, we consider a s-(1)-phenylethyl (s1pe) peptoid in an aqueous environment similar to Refs. \cite{weiser2019cgenff,wang2020reinforcement}. 
The full system consists of one biomolecule and 546 water molecules in a $(2.9 {\rm nm})^3$ dodecahedron box. The CHARMM general force field (CGenFF)  \cite{weiser2019cgenff} is used for the biomolecule and the TIP3P model \cite{jorgensen1983comparison} is used for the water molecules; see Appendix \ref{app:MD_setup} for details. 

The CVs are the three torsion angles $\omega$ spanned by atoms (C$_\alpha$, C, N, C$_\alpha$),  $\phi$, and $\psi$, where the latter two are the same as the Ala2 molecule. The FES is constructed by both the RiD and CAS methods. The setup of the RiD method is the same as  \cite{wang2020reinforcement}. For the CAS method, we use 20 walkers and set 
the inverse of the low and high temperatures to be $\kappa_l=10$ and $\kappa_h = 2$, respectively. The initial points of these walkers at each iteration are chosen in the same method as the two-dimensional problem.
The timestep $\delta t$ for the dynamics of the random walkers \eqref{equ:langevin} is set to be $0.1$. 
For each sample point, restrained MD with $\kappa=500$ is conducted for $10000$ steps to compute the mean force.

For visualization, we project the constructed FES onto a two-dimensional plane and fix the third variable.  Fig. \ref{fig:s1pe} shows the projected FES on the $\omega-\phi$ and $\phi-\psi$ plane obtained from the CAS and the RiD methods. For each projection, the reference is constructed as a 2D FES using the metadynamics \cite{laio2002escaping}.
Similar to the previous 2D case, the present CAS method yields higher accuracy with lower computational cost. 

% The computational cost and error are shown in Table \ref{tab:computational cost sipe}. We show that with less computational cost, the CES method show better accuracy on FES construction.
\begin{figure}
    \centering
    \includegraphics[width=0.8\textwidth]{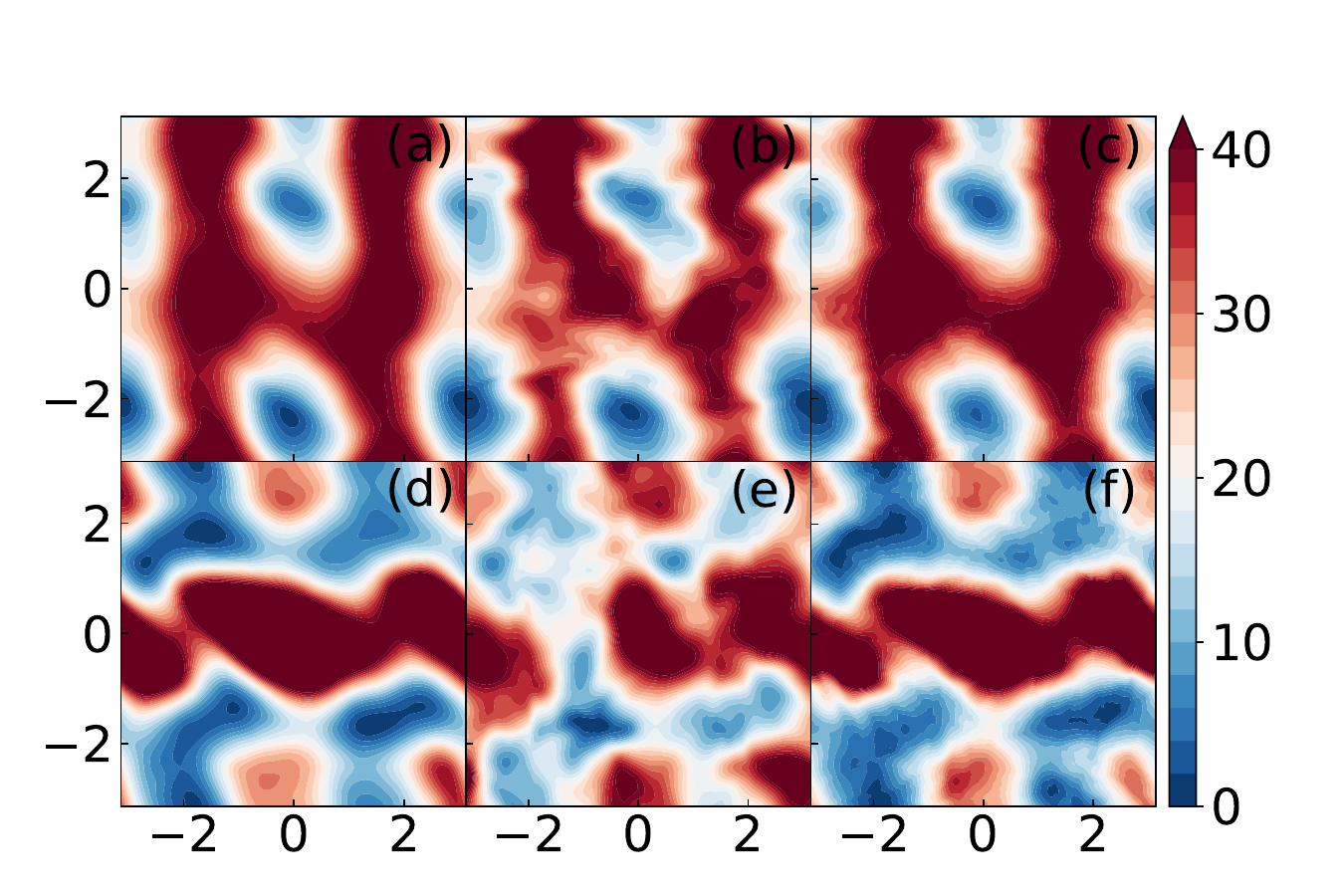}
    \caption{The 3D FES for molecule s1pe projected on the \(\omega - \phi\) plane by fixing \(\psi=1.5\) (first row) and the \(\phi -\psi\) planes by fixing \(\omega=1.5\) (second-row). (a-d) The 2D FES constructed using metadynamics with the third variable restrained (reference); (b-e) 2D projection of the 3D FES constructed by the RiD method; (c-f) 2D projection of the 3D FES constructed by the present CAS method. 
    %The CAS method requires $4.81 \times 20$ CPU hours for sampling and $0.84$ GPU hours for training, whereas the RiD method uses $423.33$ CPU hours and $8$ GPU hours, respectively. 
    %For a more detailed quantitative analysis, we refer to the Appendix.
    }
    \label{fig:s1pe}
\end{figure}

\subsection{Nine-dimensional FES}
\begin{figure}[t]%[\sidecaptionrelwidth][t!]
    \centering
    \includegraphics[width=0.8\textwidth]{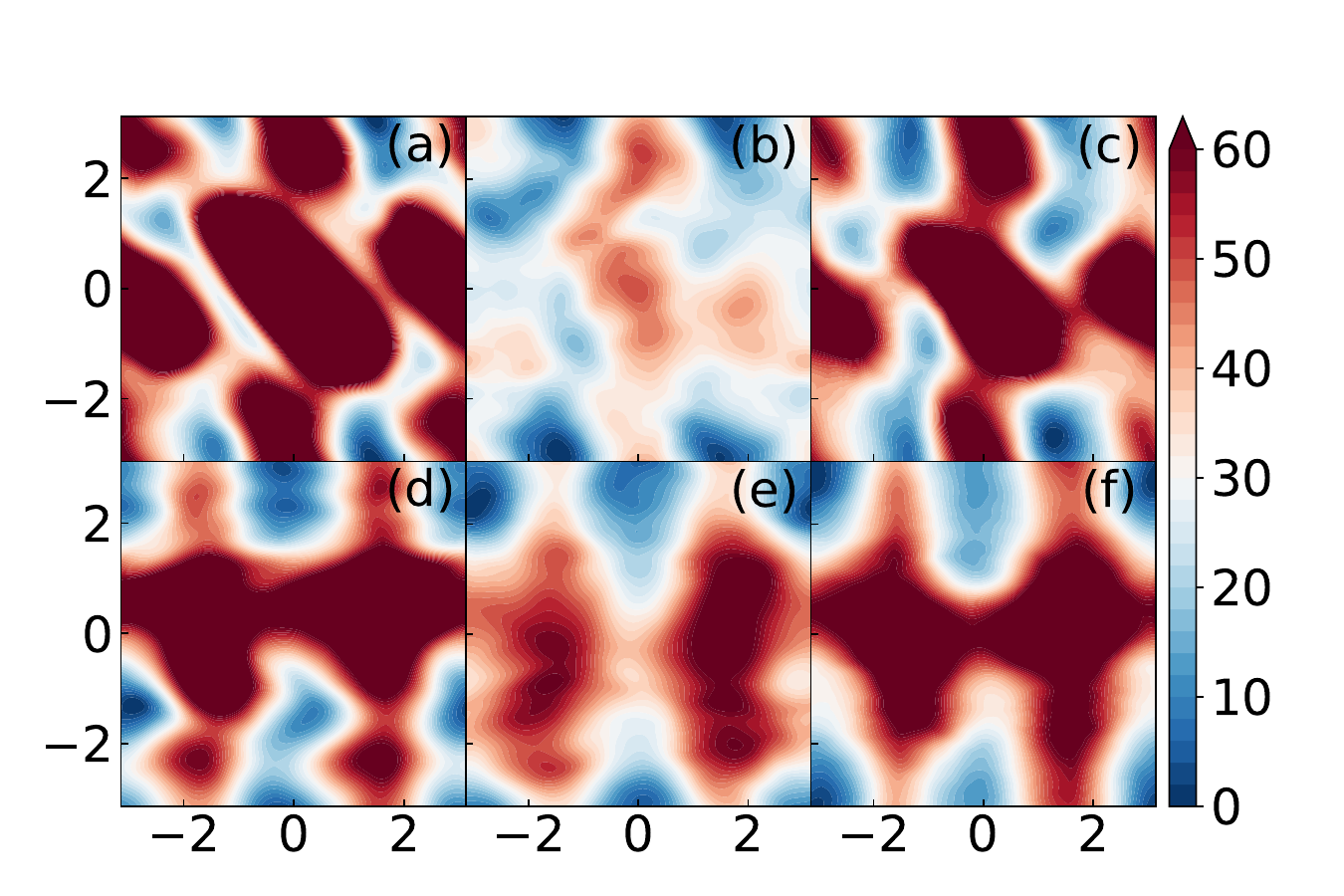}
    \caption{The 9D FES for molecule (s1pe)$_3$ projected on the $\omega_1-\phi_1$ (first row) and $\omega_1-\psi_1$ plane (second row). (a-d) The 2D FES constructed using metadynamics with the remaining variable restrained (reference solution); (b-e) 2D projection of the 9D FES constructed by the RiD method; (c-f) 2D projection of the 9D FES constructed by the present CAS method. } 
    \label{fig:s1pe3}
\end{figure}

Furthermore, we consider a more complex molecule, the peptoid trimer (s1pe)$_3$, solvated in a (4.2 nm)$^3$ dodecahedron box with 1622 TIP3P water molecules. 
%
%The force field and other simulation setups are similar to the s1pe molecule. 
%
The chosen CVs consist of the $9$ torsion angles $\omega,\phi,\psi$ that are defined in s1pe case, associated with the different C$_\alpha$ atoms and denoted as $\omega_1,\phi_1,\psi_1,\omega_2,\phi_2,\psi_2,\omega_3,\phi_3,\psi_3$. We use 64 walkers for this case and the initial conditions of the walkers at each iteration are chosen with the same method as the previous cases. The inverse of the low and the high temperature are set to be $\kappa_l = 100$ and $\kappa_h=2$, respectively. The timestep $\delta t$ of the dynamics of the random walkers is set to be $0.1$. The FES is constructed by the CAS method using 28 iterations of sampling and training, which requires $225.53\times64 = 14434.35$ CPU hours for simulation and $6.06$ GPU hours for training. For comparison, the RiD method uses $17900$ CPU hours for simulation and $15.44$ GPU hours for training. 
Similar to the above 3D problem, the constructed 9D FES is projected onto various 2D planes with the remaining variables fixed. For each 2D projection, the reference is constructed as a 2D FES using metadynamics. 
Fig. \ref{fig:s1pe3} shows the projection on the $\omega_1-\phi_1$ and $\omega_1-\psi_1$ plane (see Appendix \ref{app:9D} for other 2D projections). Compared with the Rid method, the present CAS method yields higher accuracy and lower computational cost.
%
%
% The reduction in computational cost, compared to the RiD method, partly comes from our approach eliminating the need to calculate forces precisely with prolonged restricted dynamics simulations at every sampled point. In the RiD method, inaccurately computed mean forces within the dataset can lead to different neural networks consistently making the same inaccurately predictions. Even with a small standard deviation in the neural network's output, the magnitude of the error can still be significant.
%Instead, the present method directly uses the construction error to impose the adaptivity and therefore achieve enhanced sampling of the regions lacking accuracy.

\subsection{Thirty-dimensional FES}
Finally, we consider the polyalanine-15 (Ace-(Ala)$_{15}$-Nme), denoted as Ala16, solvated in $2258$ water molecules.  %The full system is simulated temperature 300K in a ($4.62$ nm)$^3$ dodecahedron box using the Amber99-SB force field with a time step is 2 fs. We refer to the SI for the simulation details.
The chosen CVs consist of torsion angles $\phi$ and $\psi$, defined in Ala2 case, associated with the different C$_\alpha$s, denoted as $\{\phi_i,\psi_i\}_{i=1}^{15}$. We use 64 walkers for this case and the initial value of the walkers at each iteration is given by the biased simulation before. The inverse of the low and high temperatures are set to be $\kappa_l=20$ and $\kappa_h = 5$, respectively. The timestep $\alpha$ for dynamics of the random walkers Eq. \eqref{equ:langevin} is set to be $0.1$. The FES is constructed using $100$ iterations of sampling and training.  Similar to the previous case, the obtained FES is plotted on a two-dimensional plane while the remaining variables are fixed and the reference is constructed as a 2D FES using the metadynamics. Fig. \ref{fig:ala16_2D} shows the projection on the $\phi_1-\psi_1$, $\phi_2-\phi_3$ and $\psi_2-\psi_3$ plane. For all the cases, the projection of the 30-dimensional FES shows good agreement with the 2-dimensional reference solution. We have also examined the projection on other planes; the prediction shows good agreement with the reference solution as well. 
%We refer to the SI for details. 

% \begin{figure}
% \centering
% \includegraphics[width=0.8\textwidth]{./figure/ala16_trans.pdf}
% \caption{The 30D FES for molecule Ala16 projected on various 2D planes. (a-b) \(\phi_1 - \psi_1\) (c-d) \(\phi_2 - \phi_3\) (e-f) \(\psi_2 - \psi_3\). (a-c-e) 2D FES constructed by metadynamics (reference) (b-d-f) The 30D FES constructed by the present CES method. 
% }
% \label{fig:ala16_2D}
% \end{figure}

\begin{figure}%[\sidecaptionrelwidth][t!]
\centering
\includegraphics[width=0.8\textwidth]{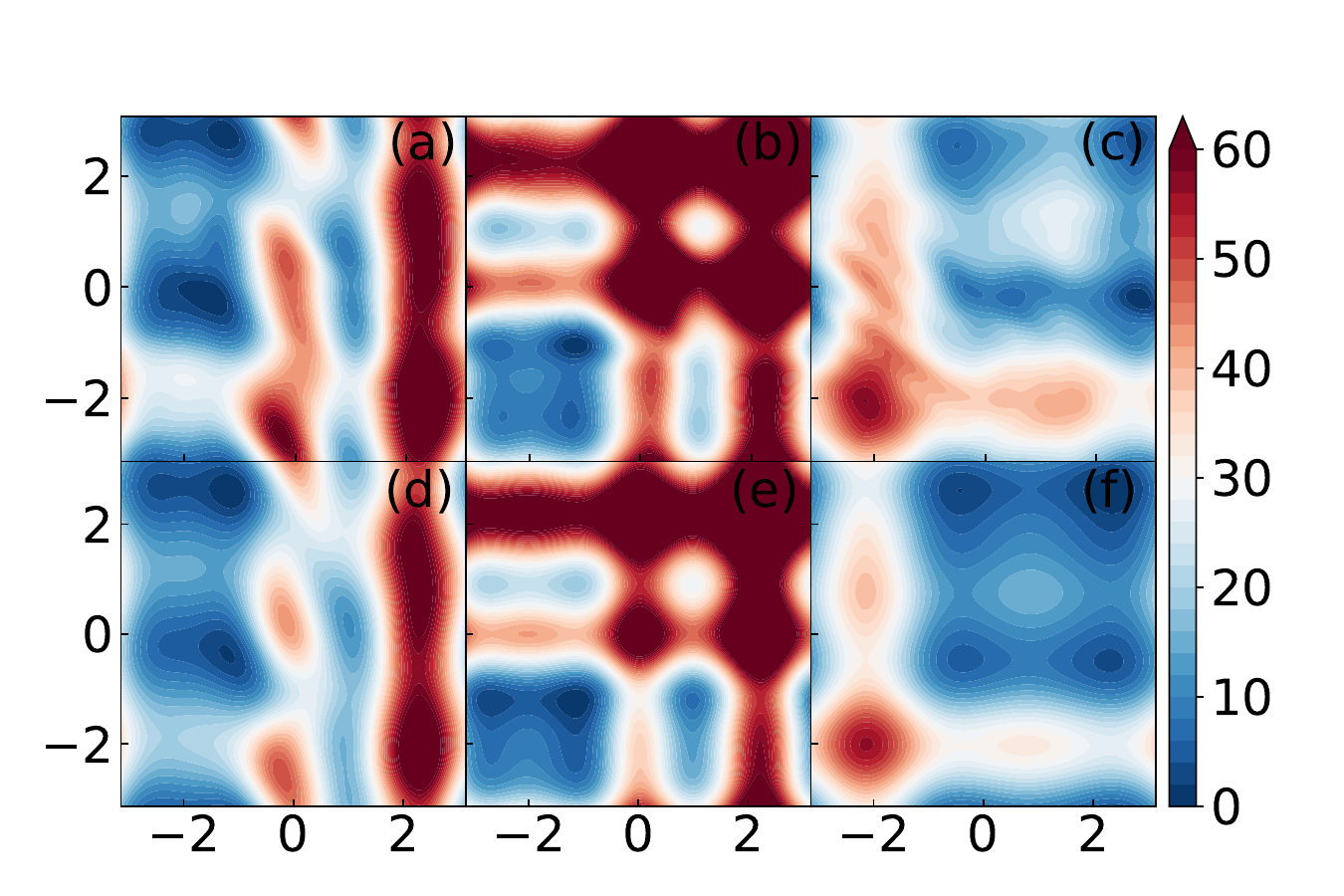}
\caption{The 30D FES for molecule Ala16 projected on various 2D planes. (a-d) \(\phi_1 - \psi_1\) (b-e) \(\phi_2 - \phi_3\) (c-f) \(\psi_2 - \psi_3\). (a-b-c) 2D FES constructed by metadynamics (reference); (d-e-f) 2D projection of the 30D FES constructed by the present CAS method. 
}
\label{fig:ala16_2D}
\end{figure}

\section{Conclusions}
We developed a consensus-based sampling approach for adaptive surrogate construction in systems with complex energy landscapes. 
The method reformulates the construction task as a minimax optimization problem by simultaneously optimizing the target function approximation and the training set. 
A unique feature of the method is its ability to unify the dynamical exploration of high-dimensional phase space in the presence of prevalent energy barriers with efficient sampling adaptive to the posterior residual error. 
As a motivating example, we considered the long-standing problem of constructing high-dimensional FESs for MD systems. Compared with existing approaches \cite{stoltz2010free} focusing on various enhanced sampling strategies to overcome energy barriers, the present method further enables the residual-based adaptive sampling that remains an open problem in FES construction. On the other hand, unlike the common adaptive sampling strategies based on the generative model \cite{tang2023adversarial} and
important sampling \cite{Gao_Yan_SIAM_2023} used for solving high-dimensional PDEs, the present method achieves efficient sampling of the max-residual region in the form of an interacting particle system that does not rely on the free query of the arbitrary point in the phase space and the gradient computation of the residual function. Given the fact the sampling only relies on the first and second-moment estimation, the method could be particularly efficient for high-dimensional problems. While the numerical results of biomolecular systems have demonstrated the effectiveness for the high-dimensional FES construction, the present method provides a framework for unifying the phase space exploration and residual-based adaptive sampling, and can be broadly applied to multiscale systems involving constrained sampling dynamics and surrogate modeling, such as rare event sampling  \cite{e2010transition}, stochastic reduced dynamics \cite{Lyu_Lei_PRL_2023, Ge_Lei_GLE_PRL_2024} and uncertainty quantification in high-dimensional constrained phase space \cite{lelievre2016partial}.

% We have presented a consensus-based approach for constructing high-dimensional FESs by reformulating the construction task as a minimax optimization problem. Rather than seeking the direct fitting, this method essentially establishes an adversarial learning of FESs by simultaneously optimizing the target function approximation and the training set. While the common approaches mainly focus on the efficient exploration of the phase space in the presence of local minima, the present method further accounts for the discretization error that has been broadly overlooked in FES construction.  Adaptive sampling of the max-residue regime is achieved through the consensus-based sampling of a posteriori residue-induced distribution in the form of a stochastic particle system in the CV space. Given the fact the sampling only relies on the first and second-moment estimation, the method could be particularly efficient for high-dimensional problems. While the numerical results of biomolecular systems have demonstrated the effectiveness for FES construction, the present framework of unifying the residual minimization and max-residue enhanced sampling is quite general for model reduction of complex systems, e.g., the stochastic reduced dynamics (e.g., see Ref. \cite{Lyu_Lei_PRL_2023,Ge_Lei_GLE_PRL_2024}) with a state-dependent memory function for multiscale physical systems.

\clearpage

\appendix

\section{Simulation setup}
\label{app:MD_setup}
All the MD simulations are
performed using the package GROMACS 2019.2 \cite{lindahl_2019_2636382} and open-source, community-developed PLUMED library \cite{tribello2014plumed}.  The simulation is carried out on Intel(R) Xeon(R) Platinum 8260 CPU.

\subsection{ala2}
The Ace-Ala-Nme (ala2) molecule is modeled by the Amber99SB force field \cite{hornak2006comparison}. The molecule is solvated in an aqueous environment with  383 TIP3P water molecules. Periodic boundary conditions are imposed along each direction. The cut-off radius of the van der Waals interaction is 0.9 nm.
The long-range Coulomb interaction is treated with the smooth particle mesh Ewald method with a real space cutoff of 0.9 nm and a reciprocal space grid spacing of 0.12 nm. The system is integrated with the leap-frog scheme at time step 2 fs. The temperature is set to be 300 K by a velocity-rescale thermostat \cite{bussi2007canonical} with a relaxation time of 0.2 ps. The Parrinello-Rahman barostat \cite{parrinello1981polymorphic} with a relaxation time scale of 1.5 ps and a compressibility
of $4.5 \times 10^{-5}$ bar$^{-1}$
is coupled to the system to control the
pressure to 1 bar. The hydrogen atom is
constrained by the LINCS algorithm \cite{hess1997lincs} and the H–O bond and H–O–H angle of water molecules are constrained by the SETTLE algorithm \cite{miyamoto1992settle}.
\subsection{s1pe}
The s-(1)-phenylethyl (s1pe) molecule is modeled by the CHARMM general force field (CGenFF) \cite{weiser2019cgenff}.
The molecule is dissolved in 546 TIP3P water molecules in a (2.69 nm)$^3$ dodecahedron box. The cut-off radius of the van der Waals interaction is 1 nm. Other setups are similar to the ala2 molecule. 
% The long-range Coulomb interaction is treated with the smooth particle mesh Ewald method with a real space cutoff of 1 nm and a reciprocal space grid spacing of 0.12 nm. The system is integrated with the leap-frog scheme at time step 2 fs. The temperature of the system is set to 298 K by a velocity-rescale thermostat \cite{bussi2007canonical} with a relaxation time of 0.2 ps.
% The Parrinello-Rahman barostat pressure couple, the hydrogen atom constraint, the H–O bond constraint, and the H–O–H angle of water molecules constraint are the same as the previous one. 

\subsection{(s1pe)$_3$}
The (s1pe)$_3$ molecule is modeled by the CHARMM general force field (CGenFF)  \cite{weiser2019cgenff}.
The molecule is solvated in a (4.2 nm)$^3$ dodecahedron
box with 1622 TIP3P water molecules. Other setups are similar to the ala2 molecule.

\subsection{ala16}
The Ace-(Ala)$_{15}$-Nme (ala16) molecule is modeled by the Amber99SB force field \cite{hornak2006comparison}. The molecule is solvated in a (4.62 nm)$^3$  dodecahedron box with 2258 TIP3P water molecules. Other setups are similar to the ala2 molecule.

\section{Training}
The training data is collected by restrained molecular dynamics.
The initial $10\%$ steps of the restrained dynamics are used as equilibrium and the rest $90\%$ steps are used to calculate the mean force.  
The FES are parameterized as a fully connected neural network. The depth and width of the NN are shown in Table \ref{tab:NN}.
The NNs are trained by Adam for $100000$ steps with a learning rate $1\times10^{-3}$. For each training step, $5000$ sampling points are randomly selected from the data set. All the training processes are conducted using an Nvidia GPU V100 with 32GB memory.
\begin{table}[h]
    \caption{The depth and width of the NN used to parameterize the FES of different molecules.}
    \centering
    \begin{tabular}{|c|c|c|}
    \hline
         molecule &  depth  & width \\
    \hline
         ala2 &  3 & 48 \\
         \hline
         s1pe &  4 & 64 \\
         \hline
         (s1pe)$_{3}$ & 4 & 512 \\
         \hline
         ala16 & 4 & 640 \\ 
    \hline
    \end{tabular}
    \label{tab:NN}
\end{table}

\section{Parameter choice and senstivity study}
\label{sec:parameter_sen}
{
\color{black}{
%\subsection{Parameter choice}
The present CAS method relies on the choice of the number of walker $N_w$, the inverse of the low temperature $\kappa_l$, and the high temperature $\kappa_h$. Specifically, we choose $N_w$ to achieve a balance between the approximation of the first and second moments of $q^{\ast}$ and the computational cost of solving the sampling dynamics. $\kappa_l$ is chosen to achieve a balance between the Laplace approximation of the max-residual region and the Lipschitz condition of the weighted density function $p(z)$. $\kappa_h$ is chosen to achieve a balance between the exploration of the full phase space and the adaptive sampling of the max-residual region. Numerical results show that the present method is robust with the choice of $N_w$ with $N_w \ge 10$ and a broad range of temperatures $\kappa_l$ and $\kappa_h$.

To examine the robustness of the present method, we report a senstivity study.
Specifically, we fix all parameters to the baseline setting in the main text and vary a single parameter. Each parameter setting is repeated with 5 random seeds. Fig.~\ref{fig:ala2_sens} shows the mean of the $L^2$ error of the 2D FES for the ala2 molecule system, with shaded bands indicating the range of the standard deviation with  $N_w\in\{5,10,20\}$, $\kappa_h\in\{2,0.2,0.02\}$ and $\kappa_l\in\{100,10,1\}$. 
We observe that the method exhibits stable convergence with $N_w \ge 10$ and a broad range of $\kappa_l$ and $\kappa_h$. Similarly, Fig. \ref{fig:chi1_sens} shows the $L^2$ error of the 3D FES for the s1pe molecule, where the method remains robust with $N_w \in \{10,20, 40\}$.
}
}

%Larger $\kappa_h$ leads to slower error reduction and higher variability, consistent with reduced barrier-crossing, while smaller $\kappa_h$ improves robustness within a practical range.

\section{Exploration of the phase space}
\label{sec:explore_CV}
{
\color{black}{
To illustrate the dynamical exploration of the phase space, Fig. \ref{fig:sampling} shows the evolution of the random walkers for the ala2 molecule system. We note that the sampling dynamics of the random walkers \eqref{equ:langevin} is driven by a non-local potential $G(\mb z)$. In particular, $G(\mb z)$ takes a \textcolor{black}{quadratic} form determined by the local residual $\mathcal{L}^{-}(\mb z)$ and does not explicitly depend on the underlying MD potential $U(\mb r)$. This feature differs from most existing sampling methods where the sampling dynamics is coupled with the stiff potential $U(\mb r)$, and therefore, enables the random walkers to efficiently explore the 
%regions separated by energy barriers 
CV space and achieve the adaptive sampling of the FES.
}}

% Performance remains stable across a wide range of $\kappa_l$, illustrating that the method is not overly sensitive to $\kappa_l$ provided the reweighting does not overly collapse the effective sample size. 
\begin{figure}[t]
    \centering
    \includegraphics[width=0.3\linewidth]{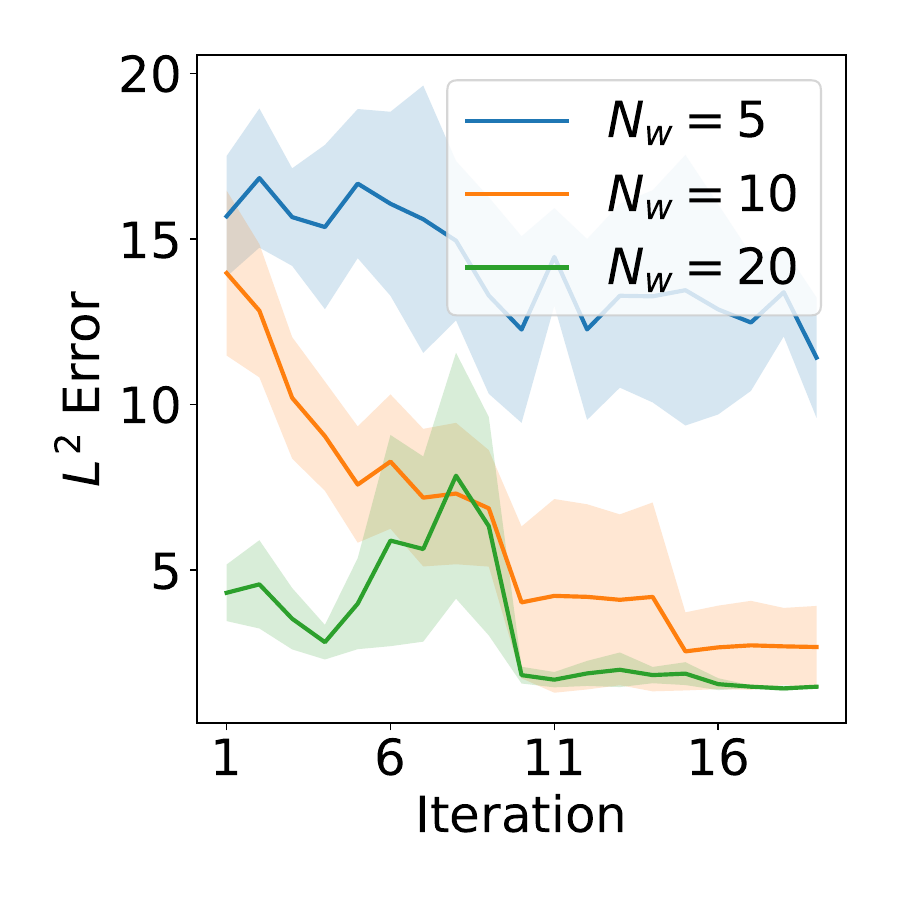}
    \includegraphics[width=0.3\linewidth]{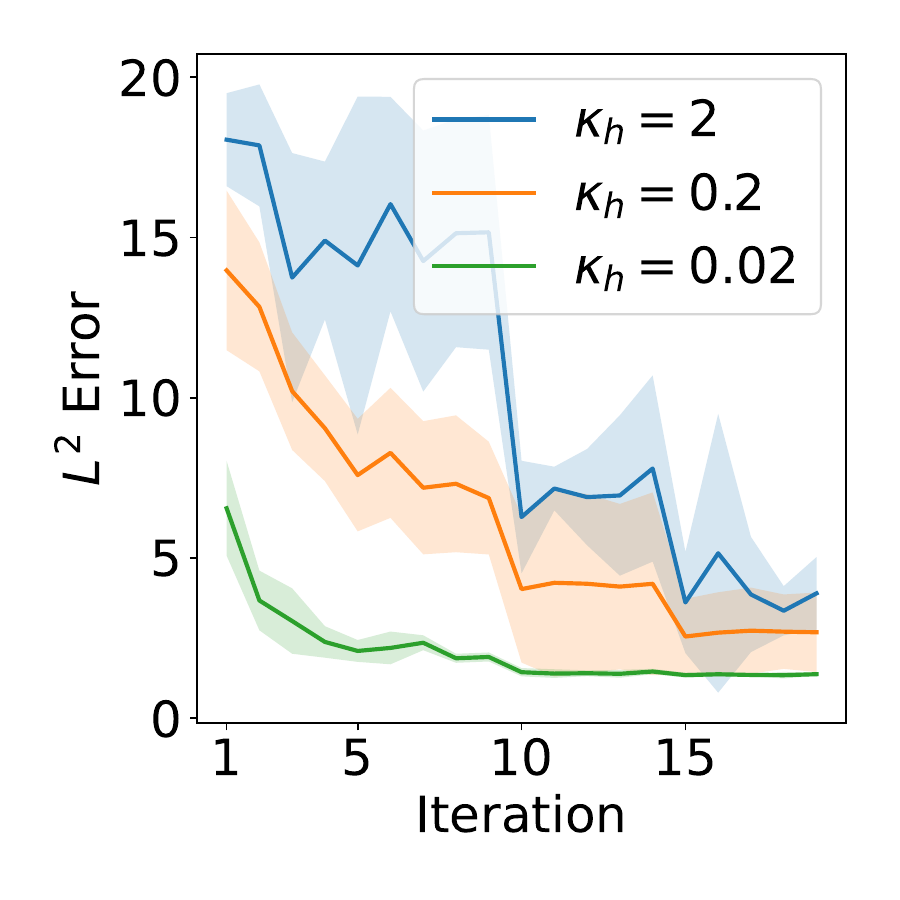}
    \includegraphics[width=0.3\linewidth]{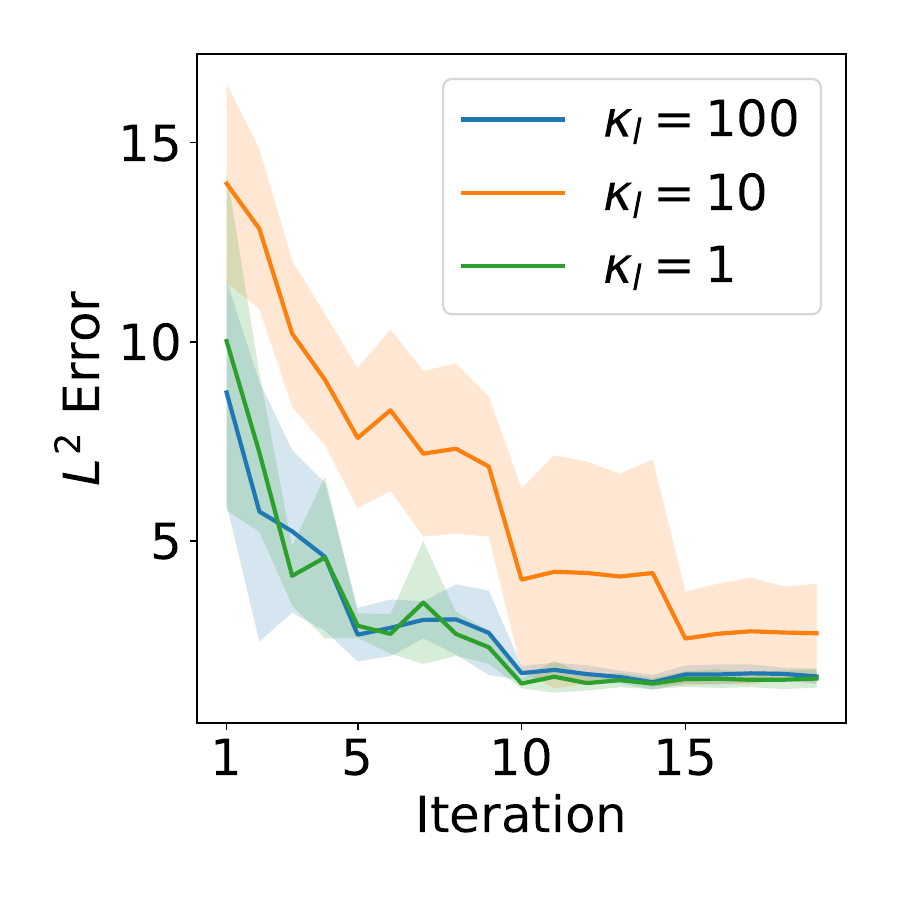}
    \caption{{\color{black}{Sensitivity study on the $L^2$ error of the 2D FES for the ala molecule with respect to the number of walkers $N_w$, the inverse temperature $\kappa_h$, and $\kappa_l$. Solid lines: mean over 5 random seeds; shaded region: one standard deviation.}}}
    \label{fig:ala2_sens}
\end{figure}
\begin{figure}
    \centering
    \includegraphics[width=0.45\linewidth]{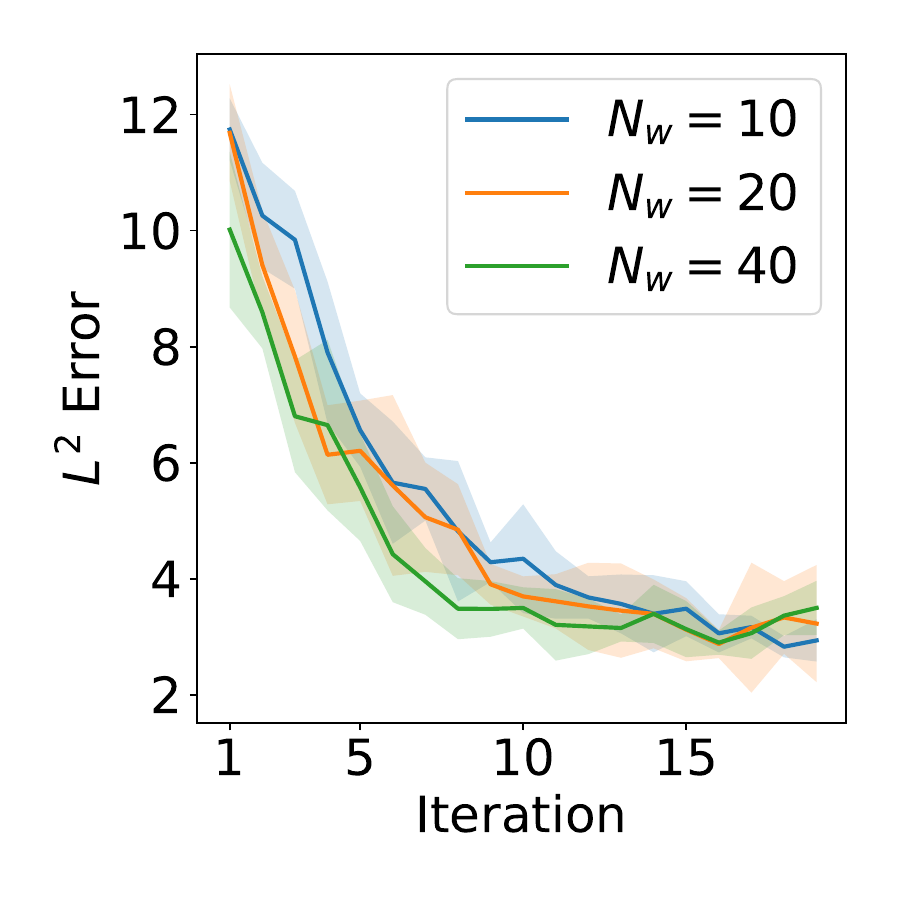}
    \caption{{\color{black}{Sensitivity study on the $L^2$ error of the 3D FES for the s1pe molecule with respect to the number of walkers $N_w$. Solid lines: mean over 5 different random seeds; shaded region: one standard deviation.}}}
    \label{fig:chi1_sens}
\end{figure}

\begin{figure}[t]
    \centering
    \includegraphics[width=0.45\linewidth]{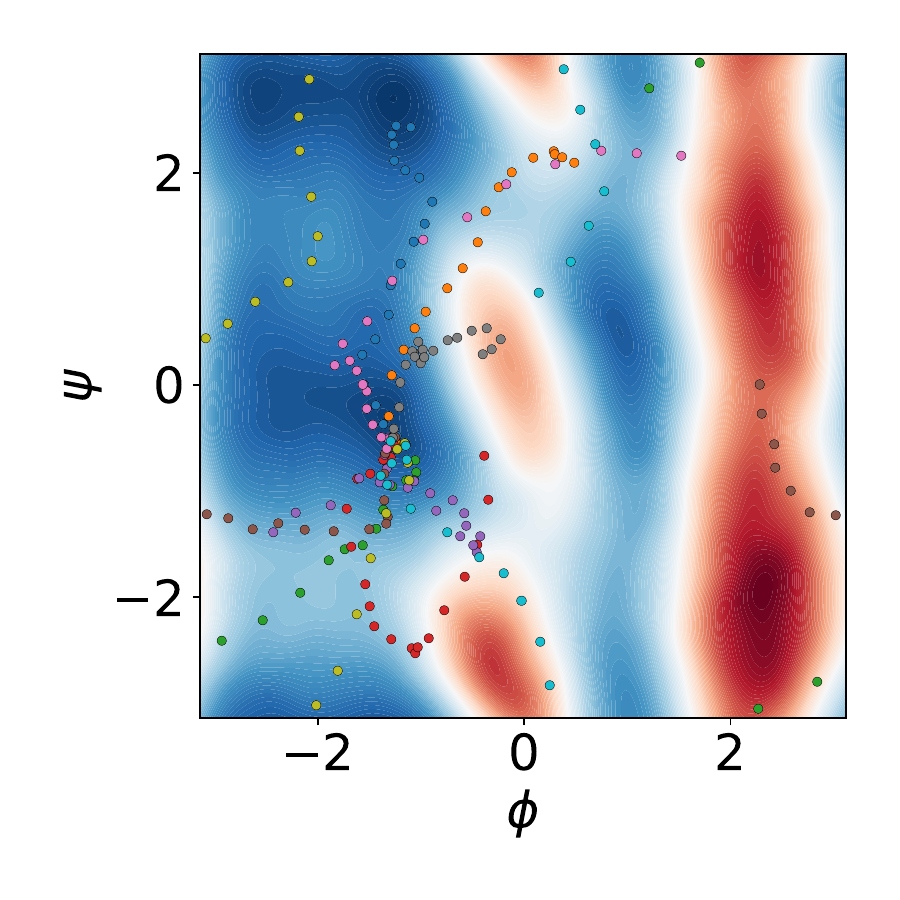}
    \caption{{\color{black}{Residual-guided exploration of the CV space of the ala2 molecule.  
    The points show the walker samples generated by the CAS method on the $(\phi,\psi)$ plane, overlaid on the reference FES.}} 
    % Left: smaller exploration parameter ($\kappa_h=0.2$), resulting in more localized sampling. 
    % Right: larger exploration parameter ($\kappa_h=2$), producing broader converage and more frequent visits to regions separated by free-energy barriers.
    }
    \label{fig:sampling}
\end{figure}

\section{Additional numerical results}
\label{sec:additional_results}
\subsection{1D Rastrigin function}
\label{app:1D}
Tab. \ref{tab:1D_appendix} presents the numerical estimations of the first and second moments along with the construction of the 1D Rastrigin function. As shown in Prop. \ref{prop:invariant_dist}, the first moment $m_i$ provides the prediction of the max-residual point $z_i^{\ast}$ for the $i\mhyphen$th iteration. The second moment $V_i$ characterizes the second derivative of the residual $\left\vert f(z) -  f_{\theta_i}(z)\right\vert$ at $z = z_i^{\ast}$.   
As shown in Tab. \ref{tab:1D_appendix}, the predictions
of $m_i$ and $f''_{\theta_{i-1}} \pm V_i^{-1}$ from the present method show good agreement with the max-residual point $z_i^{\ast}$ and the second derivative $f''(z^{\ast}_i)$  for each iteration.

% is validated through the comparison between $f''$ and $f''_{\theta_{i-1}} \pm V^{-1}$, with the sign influenced by the residual's sign. Numerical results verify the accurate prediction of the second momentum, i.e., $\left\vert f^{''}(x^{\ast}) -  f^{''}_{\theta_i}(x^{\ast})\right\vert \approx V^{-1}$ for $x^{\ast}$ at the max-residual point.  
\begin{table}[ht]
    \label{tab:1D_appendix}
    \centering
        \caption{
The max-residual point $z^*_i$ and its second derivative $f''(z^*_i)$, as well as their numerical estimations from the present sampling method at each iteration. Specifically, the max-residual point $z^*_i$ can be estimated by $m_i$ and the second derivative $f''(z^*_i)$ can be estimated by $f''_{\theta_{i}} \pm {V}^{-1}$. 
}
    \resizebox{\columnwidth}{!}{\begin{tabular}{|c|c|c|c|c|c|c|c|}
    \hline
     $i$ & 0 & 1 & 2 & 3 & 4 & 5 & 6  \\ 
      \hline
        $z^*_i$ &  0.000&2.000&-2.000&0.996&-0.996&2.805 &-2.806\\
        \hline
        $m_i$& -0.000&2.000&-2.000&0.996&-0.995&2.805&-2.806\\
        \hline
        $f''(z^*_i)$ & 41.478&41.478&41.478&41.468&41.462&15.273&15.539\\
        \hline
        $f''_{\theta_{i}} \pm {V}^{-1}$  & 40.974&37.442&36.916&38.283&39.519&11.841&13.380\\
        \hline
        $i$ & 6 &   7 & 8& 9 & 10 & 11 & 12 \\
        \hline
        $z^*_i$ & -2.806 &1.500&-1.502&0.503&-0.501&2.416&-2.413\\
        \hline
        $m_i$& -2.806 & 1.499&-1.503&0.502&-0.501&2.414&-2.413\\
        \hline
        
        $f''(z^*_i)$ & 15.539 & -37.478&-37.473&-37.474&-37.478&-31.901&-31.688\\
        \hline
        $f''_{\theta_{i}} \pm {V}_i^{-1}$  &15.380 & -36.796&-36.271&-37.515&-37.285&-31.521&-31.189\\
    \hline
    \end{tabular}}
    \label{tab:second_moment_1D_function}
\end{table}

\subsection{The 2D FES for molecule ala2}
Tab. ~\ref{tab:2d_l2_error} shows the approximation error of the constructed FES surrogate at each iteration. The total computational cost is presented in Tab. \ref{tab:computational cost} in the main manuscript. 

\begin{table}[]
    \centering
    \caption{The $l_2$ and $l_\infty$ error of the FES surrogate by the present CAS method for each sampling iteration.}
    \begin{tabular}{|c|ccccccc|}
    \hline 
    iteration & 1 & 2 & 3 & 4 & 5 & 6 & 7 \\
    \hline 
    $l_{\infty}$ error &  35.21 &   30.62& 22.38&  16.41&  14.83&  9.86&  10.68\\
    \hline
    $l_2$ error &  9.82 &  7.52& 6.84&  3.77 & 2.36 &  3.37 & 1.88   \\
    \hline
    \end{tabular}

    \label{tab:2d_l2_error}
\end{table}

\subsection{The 3D FES for molecule s1pe}
\label{app:3D}
Fig. \ref{fig:chi1_ome_-1}-\ref{fig:chi1_psi_0} shows the additional 2D projections of the 3D FES (molecule s1pe) constructed by the Rid and presented CAS method. For each case, the reference is constructed as a 2D FES using the metadynamics. 
%The CAS method yields a better agreement with the reference. 
The computational cost and accuracy of CAS and RiD are shown in Tab. \ref{tab:computational cost sipe}. It shows that the CAS method achieves better accuracy with less computational cost.
% We find the RiD will typically underestimates the FES. Although the low energy part of FES is more important than the high energy part, it does not necessary mean the high energy part can be estimated as a low energy one.

\begin{table}[h]
    \centering
    \caption{The accuracy of the constructed 3D FES (the s1pe molecule) and computational time for RiD and the present CAS method. The \( l_2 \) and \( l_\infty \) error are calculated up to \( 40 \) KJ/mol. For each case, the FES is projected onto a 2D plane with the third variable fixed; the reference solution is constructed as a 2D FES by the metadynamics.  The simulation time of the CAS method is multiplied by $20$ since 20 walkers are used.}
    \begin{tabular}{|c|c|c|c|c|c|}
\hline
 \multirow{2}{*}{Method} &\multirow{2}{*}{Restraint} &\multicolumn{2}{c|}{Accuracy} &   \multicolumn{2}{c|}{Time} \\
 \cline{3-6}
 ~ & ~ & $l_2$ error & $l_{\infty}$ error & Sampling & Train  \\
\hline
 \multirow{2}{*}{RiD} & $\psi=1.5$ & $5.76$ & $25.72$ & \multirow{2}{*}{$423.33$} & \multirow{2}{*}{$8$ (GPU)} \\
 \cline{2-4}
    ~ &  $\omega=1.5$ & $12.04$ & $49.13$ & &  \\
\hline
 \multirow{2}{*}{CAS}  &  $\psi=1.5$ & $2.44$ & $11.21$  & \multirow{2}{*}{$4.81\times20$ } & \multirow{2}{*}{$0.84$ (GPU)} \\
 \cline{2-4}
  &  $\omega=1.5$ & $3.89$ & $28.80$ & & \\
\hline
    \end{tabular}
    \label{tab:computational cost sipe}
\end{table}

\subsection{The 9D FES for molecule (s1pe)$_3$}
\label{app:9D}
The addition 2D projections of the 9D FES for molecule (s1pe)$_3$ are presented in  Fig. \ref{fig:s1pe3_add}. Similar to the 3D case, the FES constructed by the present CAS method shows a better agreement with the reference constructed as a 2D FES using metadynamics.

\subsection{The 30D FES for molecule ala16}
\label{app:30D}
The addition 2D projections of the 30D FES for molecule ala16 are presented in Fig. \ref{fig:ala16_psi5_psi5_2D}. We note that the projection on the $\phi_5-\psi_5$ plane is significantly different from other projections such as the $\phi_1-\psi_1$ plane in the main context. 
%The local minimum on the $\phi_1-\psi_1$ plane will become the local maximum on the $\phi_5-\psi_5$ plane. 
While the ala16 molecule consists of 15 sequential alanine residues, the FES for individual $\phi-\psi$ projections shows different features. The numerical results of the present CAS method show good agreement with the reference.

{\color{black}
\section{Reference FES by Metadynamics}
In this section, we list the parameters \texttt{PACE},~\texttt{HEIGHT},~\texttt{SIGMA},~\texttt{BIASFACTOR}, \texttt{GRID\_BIN} and \texttt{TEMP} used for constructing the reference FES using the PLUMED \cite{tribello2014plumed} simulations of metadynamics. 
\subsection{Ala2: 2D reference FES}
We construct the 2D reference FES using a single long well-tempered metadynamics run on the CVs
$\phi$ and $\psi$.
Gaussian hills are deposited every \texttt{PACE=100} MD steps with \texttt{HEIGHT=1.0} kJ/mol and
\texttt{SIGMA=0.2,0.2} rad, using \texttt{BIASFACTOR=10} at \texttt{TEMP=300 K}.
The bias is accumulated on a grid over $[-\pi,\pi]\times[-\pi,\pi]$ with \texttt{GRID\_BIN=300,300}.
The simulation is run with timestep $\Delta t=0.002$ ps (2 fs) for $5\times 10^6$ steps (10 ns total).

\subsection{s1pe: 2D projections of the 3D FES} 
We validate the 3D FES for s1pe using 2D projections of the CV space $(\omega_1,\phi_1,\psi_1)$.
To construct consistent 2D reference FES, we run three independent well-tempered metadynamics simulations:
(i) biasing $(\omega_1,\phi_1)$ while restraining $\psi_1$,
(ii) biasing $(\omega_1,\psi_1)$ while restraining $\phi_1$,
and (iii) biasing $(\phi_1,\psi_1)$ while restraining $\omega_1$.
In each case, the remaining CV is restrained by a harmonic potential with \texttt{KAPPA=500}.
Gaussian hills are deposited every \texttt{PACE=500} MD steps with \texttt{HEIGHT=0.5} kJ/mol and
\texttt{SIGMA=0.2,0.2} rad, using \texttt{BIASFACTOR=10} at \texttt{TEMP=300 K}.
The bias is accumulated on a grid over $[-\pi,\pi]\times[-\pi,\pi]$ with \texttt{GRID\_BIN=300,300}.
Each run is performed with timestep $\Delta t=0.002$ ps (2 fs) for $5\times 10^7$ steps (100 ns total).

\subsection{(s1pe)$_3$: 2D projections of the 9D FES}
Similar to the 3D case, we construct reference 2D FES for dihedral-pair projections by running nine independent
well-tempered metadynamics simulations.
In each simulation, we bias one dihedral pair among $(\omega_k,\phi_k)$, $(\omega_k,\psi_k)$, and $(\phi_k,\psi_k)$ for $k=1,2,3$,
and apply harmonic restraints (\texttt{KAPPA=500}) to all remaining dihedrals at fixed target values (as specified in the corresponding PLUMED inputs).
Gaussian hills are deposited every \texttt{PACE=500} MD steps with \texttt{HEIGHT=1.0} kJ/mol and
\texttt{SIGMA=0.2,0.2} rad, using \texttt{BIASFACTOR=10} at \texttt{TEMP=300 K}.
The bias is accumulated on a grid over $[-\pi,\pi]\times[-\pi,\pi]$ with \texttt{GRID\_BIN=300,300}.
Each run is performed with timestep $\Delta t=0.002$ ps (2 fs) for $2\times 10^7$ steps (40 ns total).

\subsection{ala16: 2D projections of the 30D FES}
Similar to the previous case, we construct the reference 2D FES for $(\phi_1,\psi_1)$, $(\phi_2,\phi_3)$, $(\psi_2,\psi_3)$ and $(\phi_5,\psi_5)$  by running four independent
well-tempered metadynamics simulations.
In each simulation, we bias one dihedral pair,
and apply harmonic restraints (\texttt{KAPPA=500}) to all remaining dihedrals at fixed target values (as specified in the corresponding PLUMED inputs).
Gaussian hills are deposited every \texttt{PACE=500} MD steps with \texttt{HEIGHT=1.0} kJ/mol and
\texttt{SIGMA=0.2,0.2} rad, using \texttt{BIASFACTOR=10} at \texttt{TEMP=300 K}.
The bias is accumulated on a grid over $[-\pi,\pi]\times[-\pi,\pi]$ with \texttt{GRID\_BIN=300,300}.
Each run is performed with timestep $\Delta t=0.002$ ps (2 fs) for $5\times 10^7$ steps (100 ns total).
}

{\color{black}{
\section{Code availablity}
The Python implementation of the proposed method is available at 
\url{https://github.com/Lyuliyao/consensus-sampling-method-for-expolering-high-dimensional-energy-surface}. 
All reproducibility materials, including code and data used to generate the numerical results in this paper, 
are archived at \url{https://doi.org/10.5281/zenodo.19039932}.
}
}

\begin{figure}
    \centering
    \includegraphics[width=.5\textwidth]{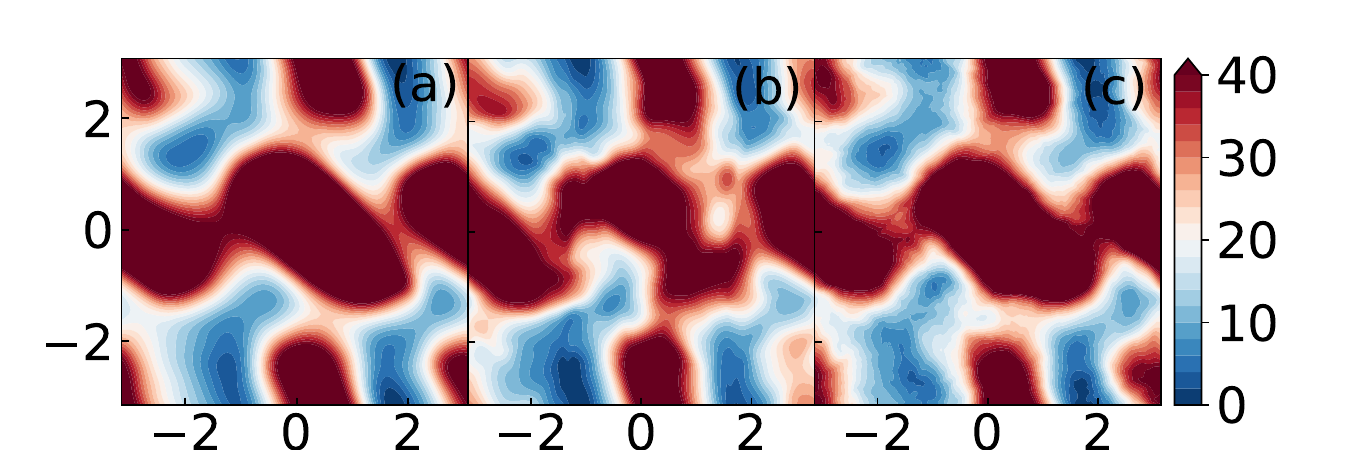}
    \includegraphics[width=.5\textwidth]{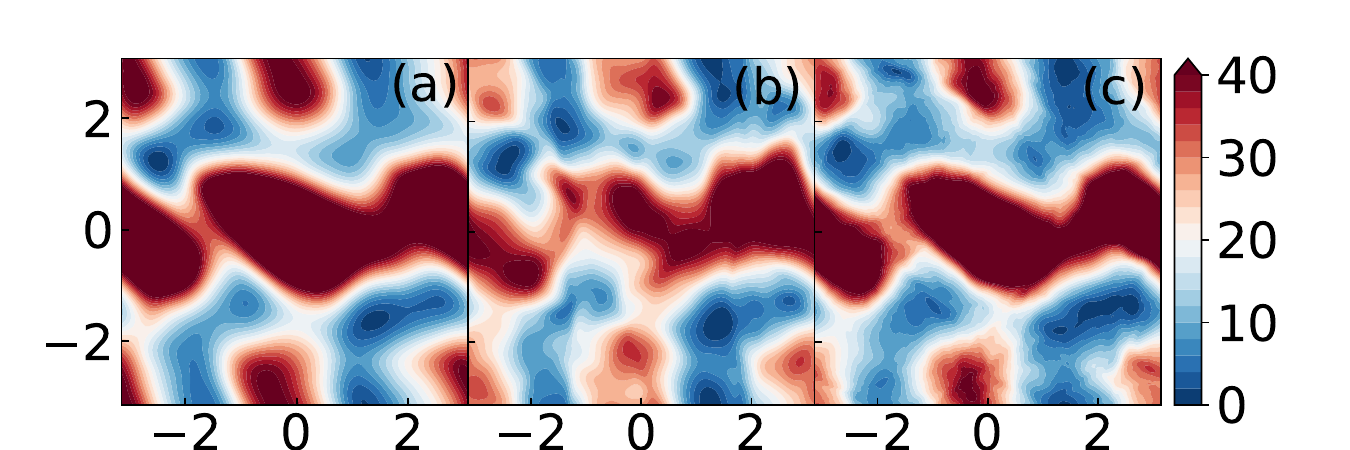}
   \caption{The 3D FES for the s1pe molecule projected on the $\phi-\psi$ plane with the third variable $\omega=-1$ 
 (first row) and $\omega=-2$ (second row). (a) 2D FES by metadynamics (reference) (b) RiD (c) CAS.}
    \label{fig:chi1_ome_-1}
\end{figure}

\begin{figure}
    \centering
    \includegraphics[width=.5\textwidth]{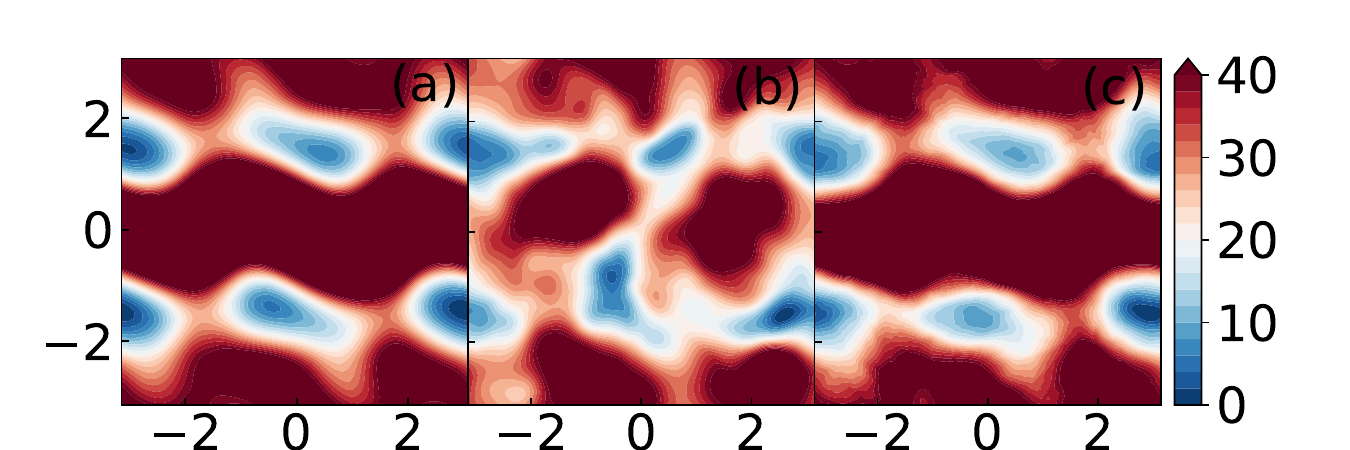}
    \includegraphics[width=.5\textwidth]{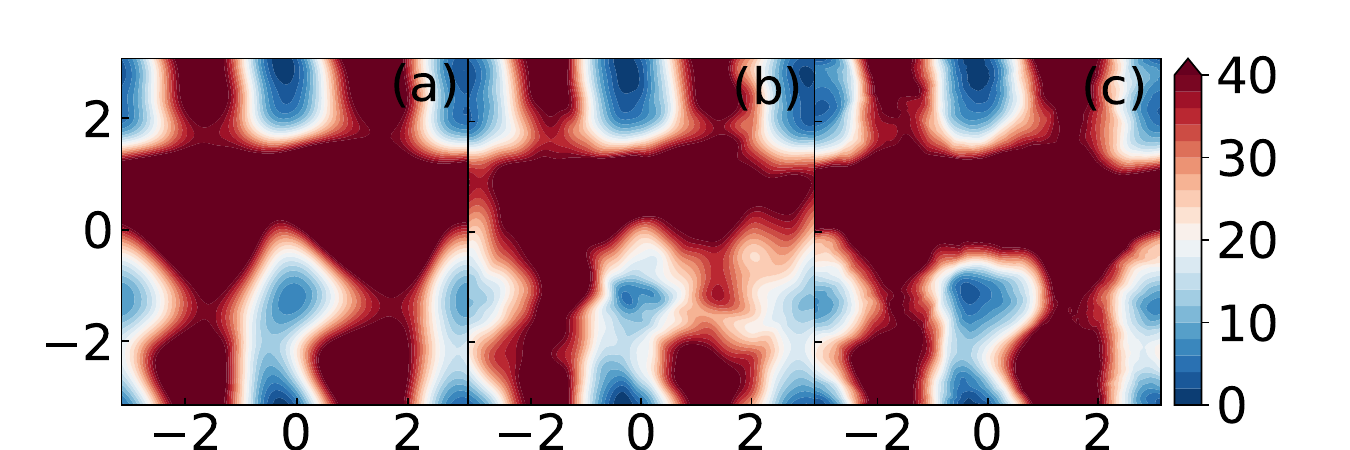}
    \caption{The 3D FES of the s1pe molecule on the $\omega-\psi$ plane with the third variable $\phi=0$ (first row) and $\phi=-1$ (second row). (a) 2D FES by metadynamics (reference)  (b) RiD (c) CAS.}
    \label{fig:chi1_phi_0}
\end{figure}

\begin{figure}
    \centering
    \includegraphics[width=.5\textwidth]{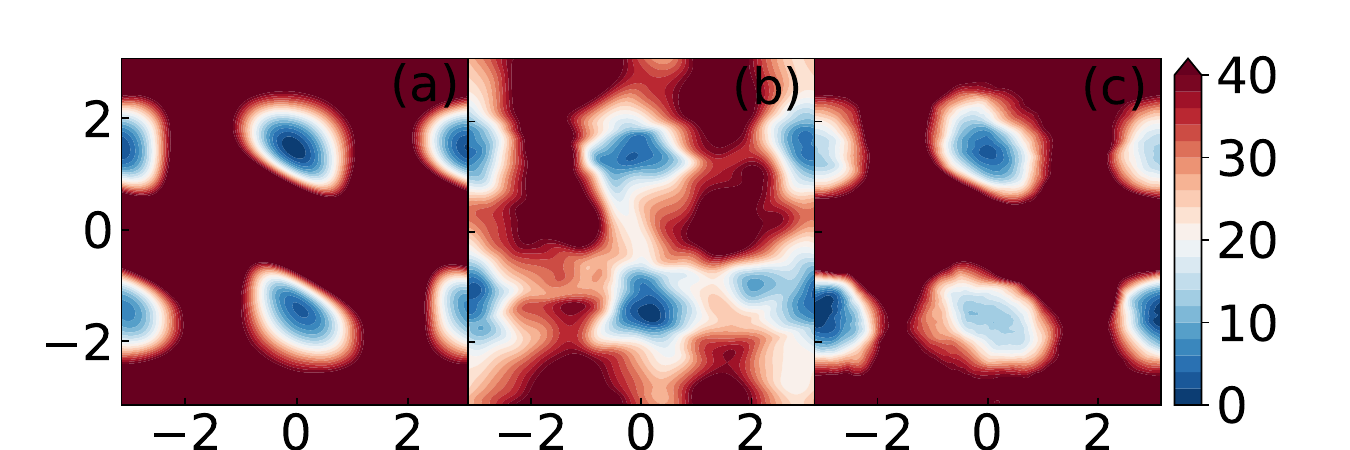}
    \includegraphics[width=.5\textwidth]{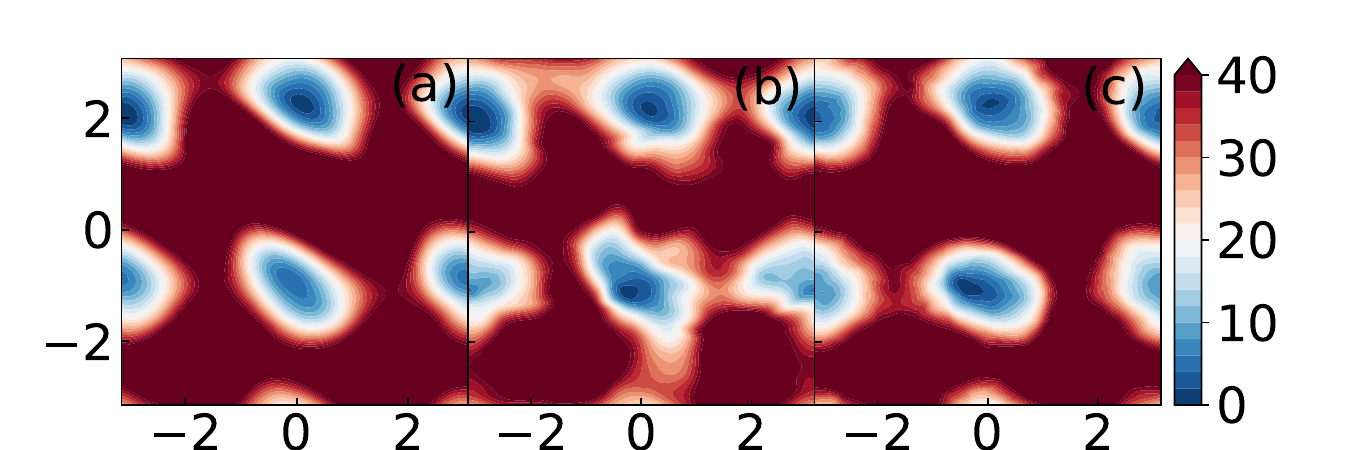}
    \caption{The 3D FES of the s1pe molecule on the $\omega-\phi$ plane with the third variable $\psi=0$ (first row) and $\psi=-1$ (second row). (a) 2D FES by metadynamics (reference)  (b) RiD (c) CAS.}
    \label{fig:chi1_psi_0}
\end{figure}

\begin{figure}
    \centering
\includegraphics[width=.5\textwidth]{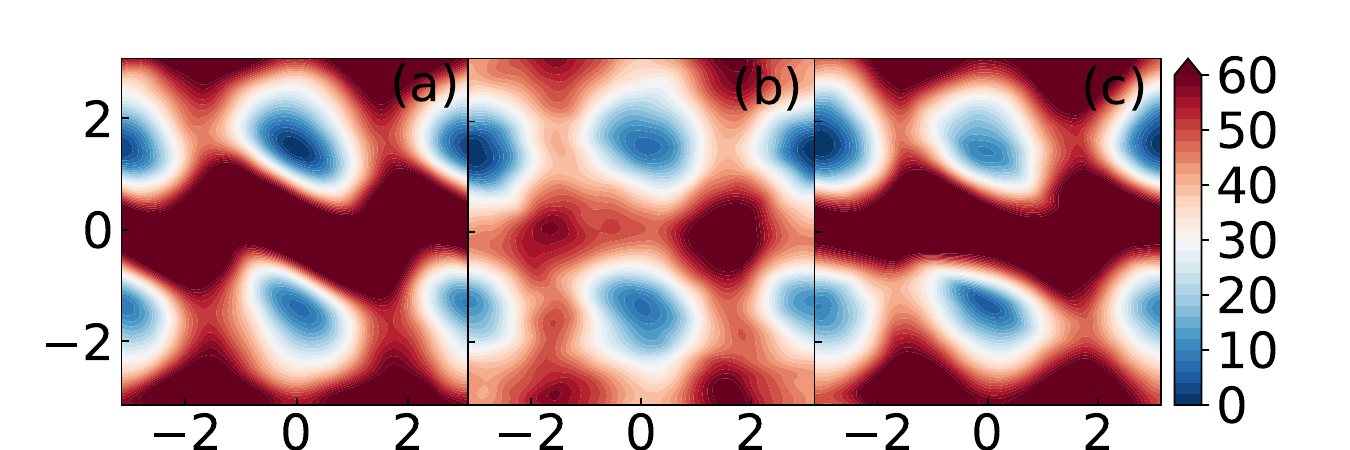}
\includegraphics[width=.5\textwidth]{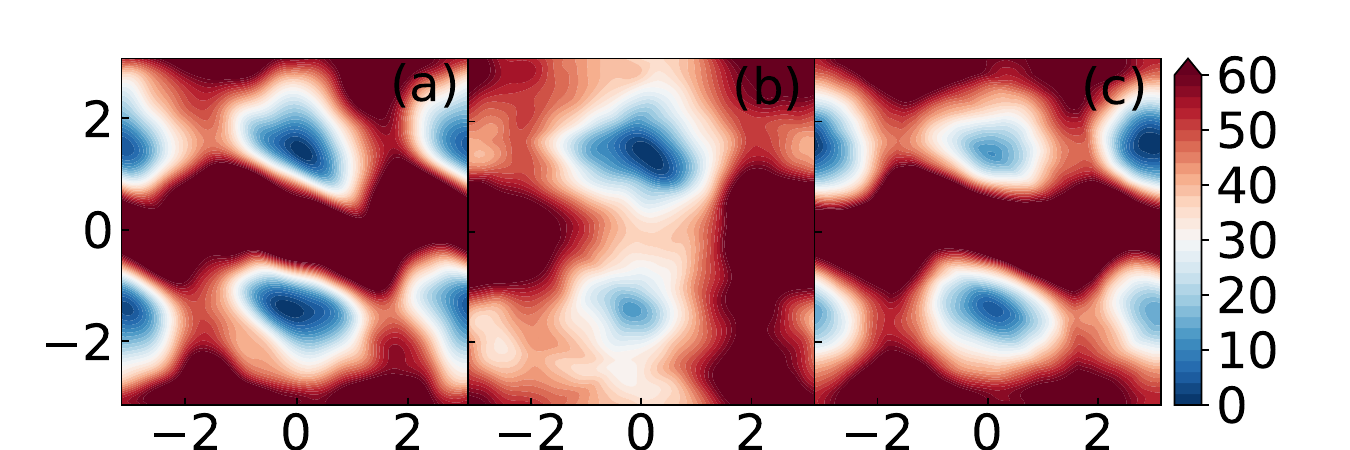}
\includegraphics[width=.5\textwidth]{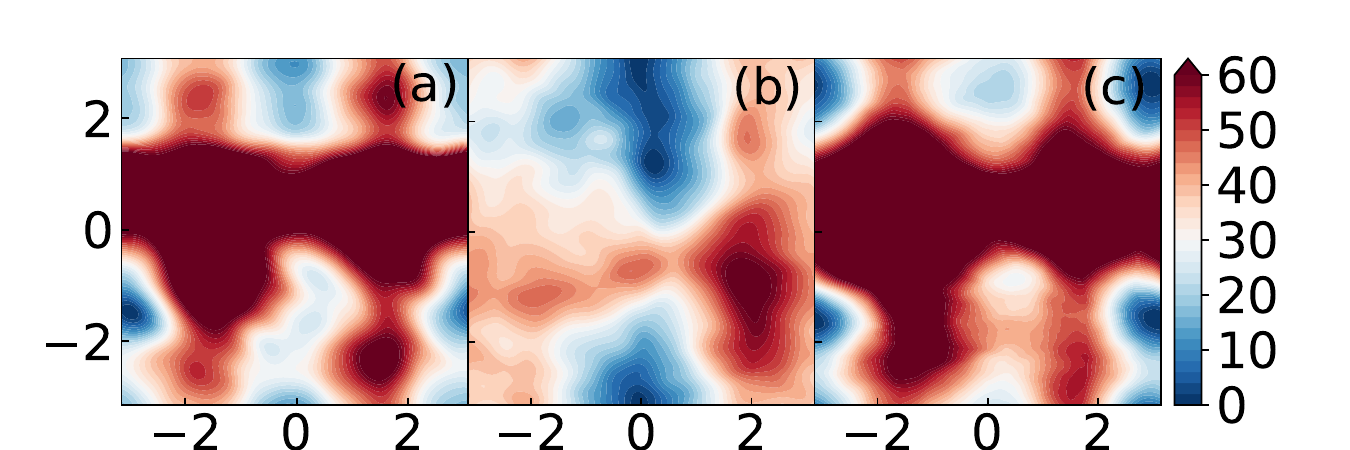}
\includegraphics[width=.5\textwidth]{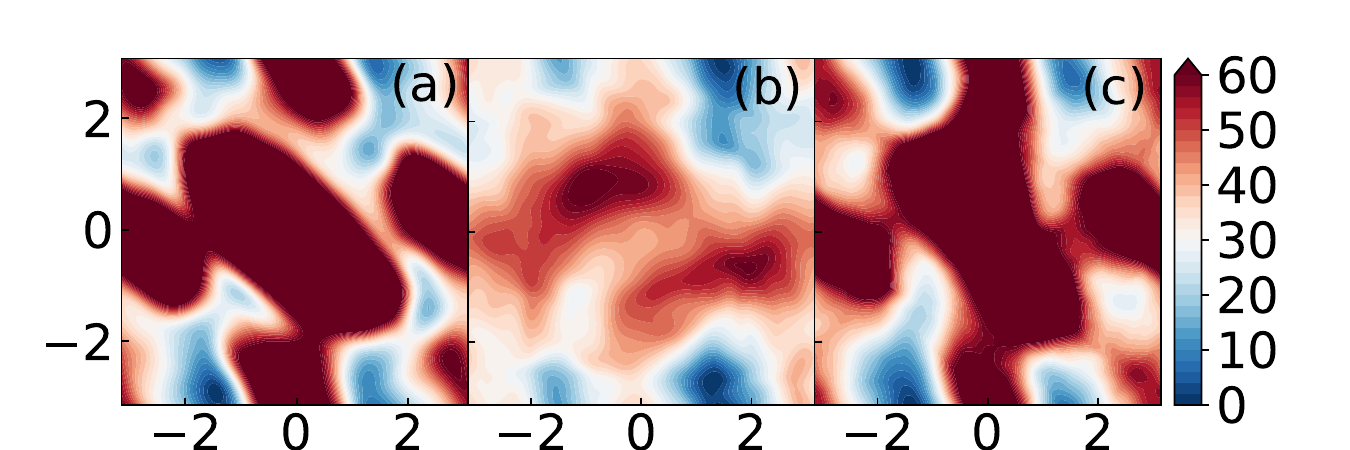}
\includegraphics[width=.5\textwidth]{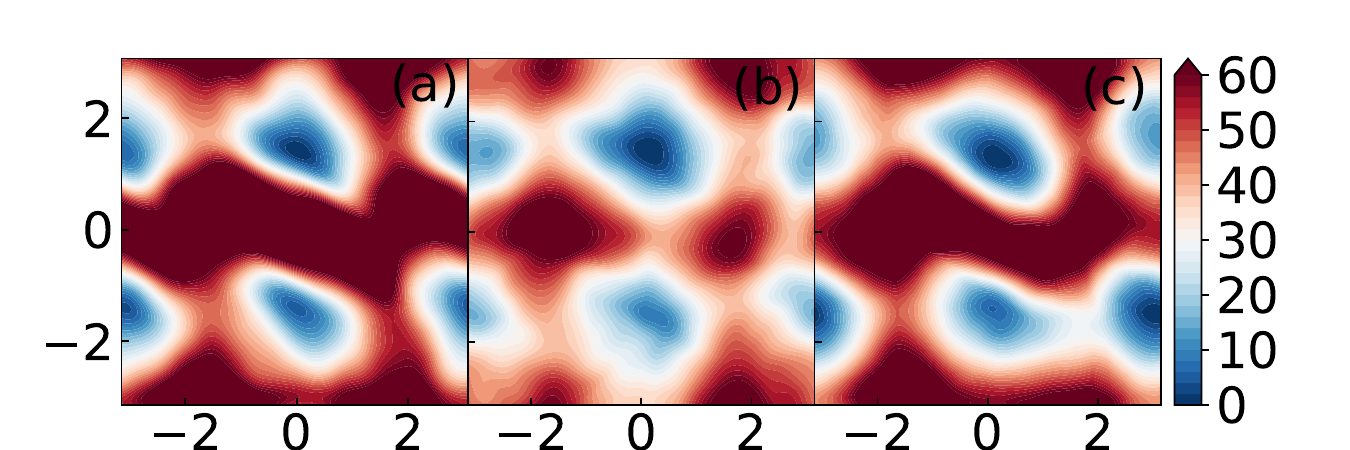}
\includegraphics[width=.5\textwidth]{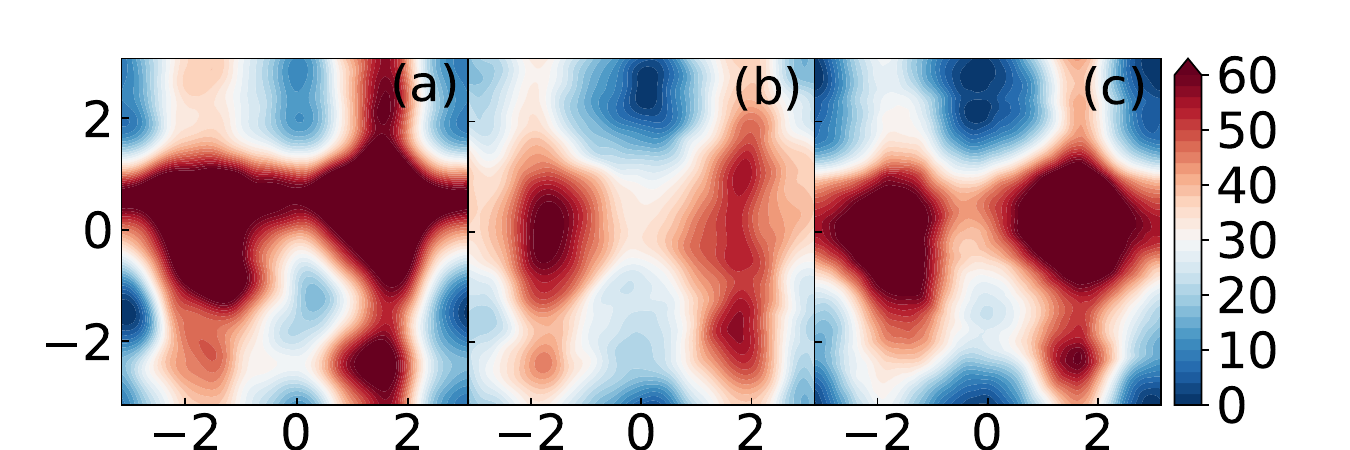}
\includegraphics[width=.5\textwidth]{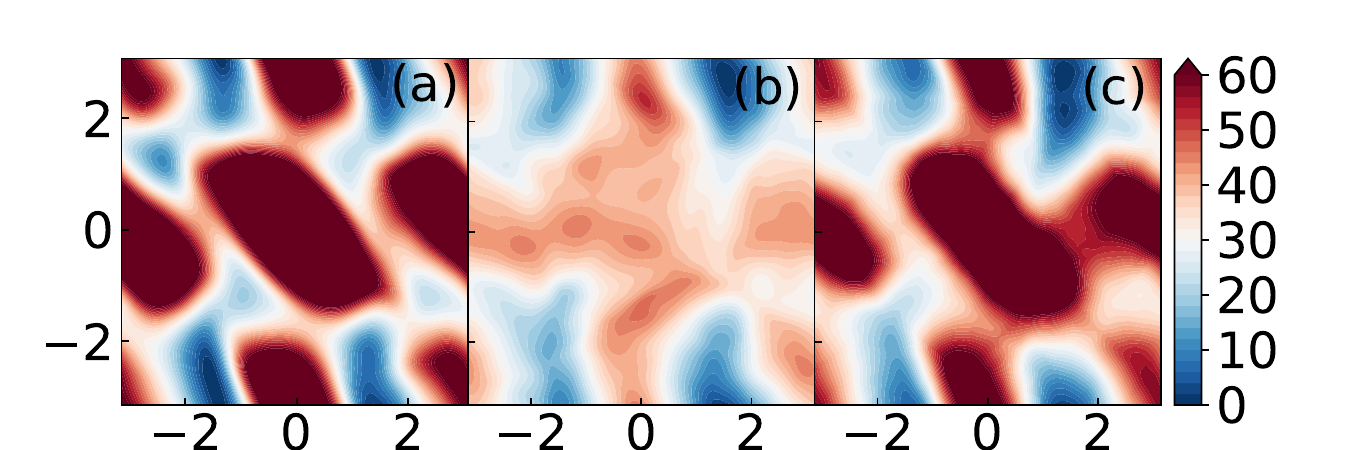}
    \caption{The 9D FES for the (s1pe)$_3$ molecule projected on the $\omega_1-\phi_1,\omega_2-\phi_2,\omega_2-\psi_2,\phi_2-\psi_2,\omega_3-\phi_3,\omega_3-\psi_3,\phi_3-\psi_3$ plane from top to bottom. (a) 2D FES by metadynamics (reference) (b) RiD (c) CAS.}
    \label{fig:s1pe3_add}
\end{figure}

\begin{figure}
    \centering
    \includegraphics[width=.5\textwidth]{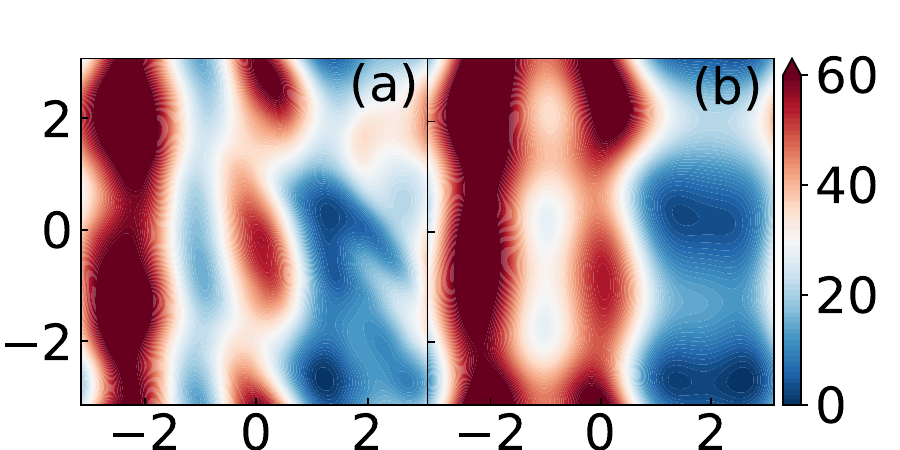}
    \caption{The 30D FES of the ala16 projected on the $\phi_5-\psi_5$ plane. (a) 2D FES constructed by metadynamics (reference) (b) 30D FES constructed by the present CAS method.}
    \label{fig:ala16_psi5_psi5_2D}
\end{figure}

\section*{Acknowledgments}
The work is supported in part by the National Science Foundation under Grant DMS-2110981 and DMS-2143739, and the ACCESS program through allocation MTH210005.

\clearpage
%\bibliographystyle{siamplain}
%\bibliography{main}

\end{document}